%% file: document.tex
\documentclass[10pt,conference]{IEEEtran}
\usepackage{etex} % http://tex.stackexchange.com/questions/38607/no-room-for-a-new-dimen
\usepackage{graphicx,subfigure,wrapfig}
\usepackage{amsmath,amssymb,alg,url,clrscode,dsfont}
\usepackage[table]{xcolor}
\usepackage{booktabs,tikz,array,multirow,bigdelim}
\usepackage{pbox,calc,dcolumn,bigstrut,comment,tabularx}
\usepackage[flushleft]{threeparttable}
\usepackage{framed}

\usepackage{adjustbox}

\usetikzlibrary{fit,shapes,arrows,backgrounds}
\usepackage{tikz-qtree}

\usepackage{listings}
\usepackage{dcolumn,multicol} 
\usepackage{pgfplots}
\pgfplotsset{compat=1.5}

\newcommand{\useMarkHighlights}{1}

\usepackage{todonotes}

\usepackage{soul,xspace}
\definecolor{colorTODO}{RGB}{255,140,255}
\definecolor{colorCHANGED}{RGB}{140,255,255}

\ifthenelse{\isundefined{\useMarkHighlights}}{

}{

}

\newcolumntype{L}{>{$}l<{$}}
\newcolumntype{P}[1]{>{\centering\arraybackslash$}p{#1}<{$}}

\newcolumntype{C}[1]{>{\centering\arraybackslash}p{#1}}

\newcolumntype{M}[1]{>{\centering\arraybackslash}m{#1}}

\newcommand\Tstrut{\rule{0pt}{2.6ex}}         % = `top' strut
\newcommand\Bstrut{\rule[-0.9ex]{0pt}{0pt}}   % = `bottom' strut

\newtheorem{definition}{Definition}

\newtheorem{theorem}{Theorem}
\newtheorem{thmdef}[theorem]{Definition}
\newtheorem{lemma}[theorem]{Lemma}

\numberwithin{theorem}{section}

\newcommand{\thmend}{\strut\hfill\ensuremath{\lrcorner}}

\newenvironment{define}[2]
{\begin{thmdef}[#1]\label{#2}$\quad$\newline\normalfont}
{\thmend\end{thmdef}}

\newcommand{\algfontsize}{\scriptsize}

\ifthenelse{\isundefined{\algmarginwidth}}{
    \setlength{\algleftmarginwidth}{.1in}
    \setlength{\algrightmarginwidth}{.1in}
}{
    \setlength{\algmarginwidth}{.1in}
}
\newcommand{\algwhiledo}[2]{\textbf{while} #1 \textbf{do} #2\\}

\definecolor{ELsel}{gray}{0.9}

\RequirePackage{pgffor}

\let\existstemp\exists
\let\foralltemp\forall
\renewcommand*{\exists}{\existstemp\mkern2mu}
\renewcommand*{\forall}{\foralltemp\mkern2mu}

\newcommand{\bag}[1]{\mathcal{B}({#1})} % Bag
\newcommand{\set}[1]{\left\{\, {#1} \,\right\}}
\newcommand{\mathlist}[1]{\left[\, {#1} \,\right]}
\newcommand{\mset}[1]{\mathlist{#1}}
\newcommand{\seq}[1]{\left\langle\, {#1} \,\right\rangle}

\newcommand{\setc}[2]{\left\{\, {#1} \;\middle|\; {#2} \,\right\}}
\newcommand{\listc}[2]{\left[\, {#1} \;\middle|\; {#2} \,\right]}
\newcommand{\msetc}[2]{\listc{#1}{#2}}
\newcommand{\seqc}[2]{\left\langle\, {#1} \;\middle|\; {#2} \,\right\rangle}

\newcommand{\seqcat}{\mathbin{\cdot}}

\newcommand{\seqshuffle}{\mathbin{\diamond}}

\newcommand{\setseq}[1]{{#1}^*}

\newcommand{\len}[1]{\left|{#1}\right|}

\newcommand{\project}[2]{#1\mathbin{\upharpoonright}_{#2}}
\newcommand{\seqproject}[2]{#1\mathbin{\upharpoonright}^*_{#2}}

\newcommand{\limpl}{\Rightarrow}

\newcommand{\funn}[1]{\mathit{#1}}
\newcommand{\fun}[2]{\mathit{#1}({#2})}

\newcommand{\lsep}{,\;}
\def\mklist#1{%
 \gdef\firstelement{1}
 \foreach \e in {#1}{%
   \ifnum\firstelement=0\lsep\fi\e%
   \gdef\firstelement{0}%
 }
}

\newcommand{\actAlph}{\mathds{A}}

\newcommand{\symTrace}{t}
\newcommand{\symEventLog}{L}

\newcommand{\symEmptyTrace}{\varepsilon}

\newcommand{\actSilent}{\tau}

\newcommand{\symTree}{P}
\newcommand{\treeLan}{\mathcal{L}}

\newcommand{\treeSetOp}{\bigotimes}
\newcommand{\treeOp}{\mathord{\otimes}}
\newcommand{\treeOpExChoice}{\mathord{\times}}
\newcommand{\treeOpXor}{\treeOpExChoice}
\newcommand{\treeOpSeq}{\mathord{\rightarrow}}
\newcommand{\treeOpLoop}{\mathord{\circlearrowleft}}
\newcommand{\treeOpParallel}{\mathord{\wedge}}
\newcommand{\treeOpPar}{\treeOpParallel}

\newcommand{\treeOpCompOr}{\mathord{\triangledown}}
\newcommand{\treeOpRecurOr}{\mathord{\vartriangle}}

\newcommand{\treeOpPlaceholder}{\mathord{?}}

\newcommand{\cf}[1]{{\scriptsize\texttt{#1}}}

\newcommand{\symContextPath}{\mathrm{C}}

\title{Recursion Aware Modeling and Discovery\\ For Hierarchical Software Event Log Analysis (Ext.)\\
\vspace*{-0.2em}
{\Large Technical Report version with guarantee proofs for the discovery algorithms}
\vspace*{-0.2em}}

\author{
\IEEEauthorblockN{Maikel Leemans}
\IEEEauthorblockA{
Eindhoven University of Technology\\
Eindhoven, The Netherlands\\
Email: m.leemans@tue.nl
\vspace*{-2em}} \and 
\IEEEauthorblockN{Wil M.
P.
van der Aalst}
\IEEEauthorblockA{
Eindhoven University of Technology\\
Eindhoven, The Netherlands\\
Email: W.M.P.v.d.Aalst@tue.nl
\vspace*{-2em}} \and 
\IEEEauthorblockN{Mark G.
J.
van den Brand}
\IEEEauthorblockA{
Eindhoven University of Technology\\
Eindhoven, The Netherlands\\
Email: M.G.J.v.d.Brand@tue.nl
\vspace*{-2em}}
}

\begin{document}

\maketitle

\begin{abstract}
This extended paper presents 1) a novel hierarchy and recursion extension to the process tree model; and 2) the first, recursion aware process model discovery technique that leverages hierarchical information in event logs, typically available for software systems.
This technique allows us to analyze the operational processes of software systems under real-life conditions at multiple levels of granularity.
The work can be positioned in-between reverse engineering and process mining.
An implementation of the proposed approach is available as a ProM plugin.
Experimental results based on real-life (software) event logs demonstrate the feasibility and usefulness of the approach and show the huge potential to speed up discovery by exploiting the available hierarchy.
\end{abstract}

\begin{IEEEkeywords}
Reverse Engineering; Process Mining; Recursion Aware Discovery; Event Log; Hierarchical Event Log; Process Discovery; Hierarchical Discovery; Hierarchical Modeling
\end{IEEEkeywords}

% -- Content - Introduction 2 --
\section{Introduction}

System comprehension, analysis, maintenance, and evolution are largely dependent on information regarding the structure, behavior, operation, and usage of software systems.
To understand the operation and usage of a system, one has to observe and study the system ``on the run'', in its natural, real-life production environment.
To understand and maintain the (legacy) behavior when design and documentation are missing or outdated, one can observe and study the system in a controlled environment using, for example, testing techniques.
In both cases, advanced algorithms and tools are needed to support a model driven reverse engineering and analysis of the behavior, operation, and usage.
Such tools should be able to support the analysis of performance (timing), frequency (usage), conformance and reliability in the context of a behavioral backbone model that is expressive, precise and fits the actual system.
This way, one obtains a reliable and accurate understanding of the behavior, operation, and usage of the system, both at a high-level and a fine-grained level.

The above design criteria make \emph{process mining} a good candidate for the analysis of the actual \emph{software behavior}.
Process mining techniques provide a powerful and mature way to discover formal process models and analyze and improve these processes based on \emph{event log} data from the system~\cite{processmining}.
Event logs show the actual behavior of the system, and could be obtained in various ways, like, for example, instrumentation techniques.
Numerous state of the art process mining techniques are readily available and can be used and combined through the Process Mining Toolkit ProM~\cite{verbeek2011xes}.
%In addition, event logs are backed by the IEEE standardized exchange format XES \cite{gunther2014xes, verbeek2011xes}.
In addition, event logs are backed by the IEEE XES standard \cite{gunther2014xes, verbeek2011xes}.

Typically, the run-time behavior of a system is large and complex.
Current techniques usually produce flat models that are not expressive enough to master this complexity and are often difficult to understand.
Especially in the case of software systems, there is often a hierarchical, possibly recursive, structure implicitly reflected by the behavior and event logs.
This hierarchical structure can be made explicit and should be used to aid model discovery and further analysis.

In this paper, \emph{we 1) propose a novel hierarchy and recursion extension to the process tree model; and 2) define the first, recursion aware process model discovery technique that leverages hierarchical information in event logs, typically available for software systems.}
This technique allows us to analyze the operational processes of software systems under real-life conditions at multiple levels of granularity.
In addition, \emph{the proposed technique has a huge potential to speed up discovery by exploiting the available hierarchy}.
An implementation  of the proposed algorithms is made available via the \emph{Statechart} plugin for ProM~\cite{prom:statechart}.
The Statechart workbench provides an intuitive way to discover, explore and analyze hierarchical behavior, integrates with existing ProM plugins and links back to the source code in Eclipse.

This paper is organized as follows (see Figure~\ref{fig:outline}).
Section~\ref{sec:related-work} positions the work in existing literature.
Section~\ref{sec:definitions} presents formal definitions of the input (event logs) and the proposed novel hierarchical process trees.
In Section~\ref{sec:heuristics}, we discuss how to obtain an explicit hierarchical structure.
Two proposed novel, hierarchical process model discovery techniques are explained in Section~\ref{sec:discovery}.
In Section~\ref{sec:postprocess}, we show how to filter, annotate, and visualize our hierarchical process trees.
The approach is evaluated in Section~\ref{sec:evaluation} using experiments and a small demo.
Section~\ref{sec:conclusion} concludes the paper.

\input{gfx/relatedwork/tbl-compare.tex}
\begin{figure}[!htb]%
    \centering%
    \vspace*{-0.5em}
    \includegraphics[width=0.48\textwidth]{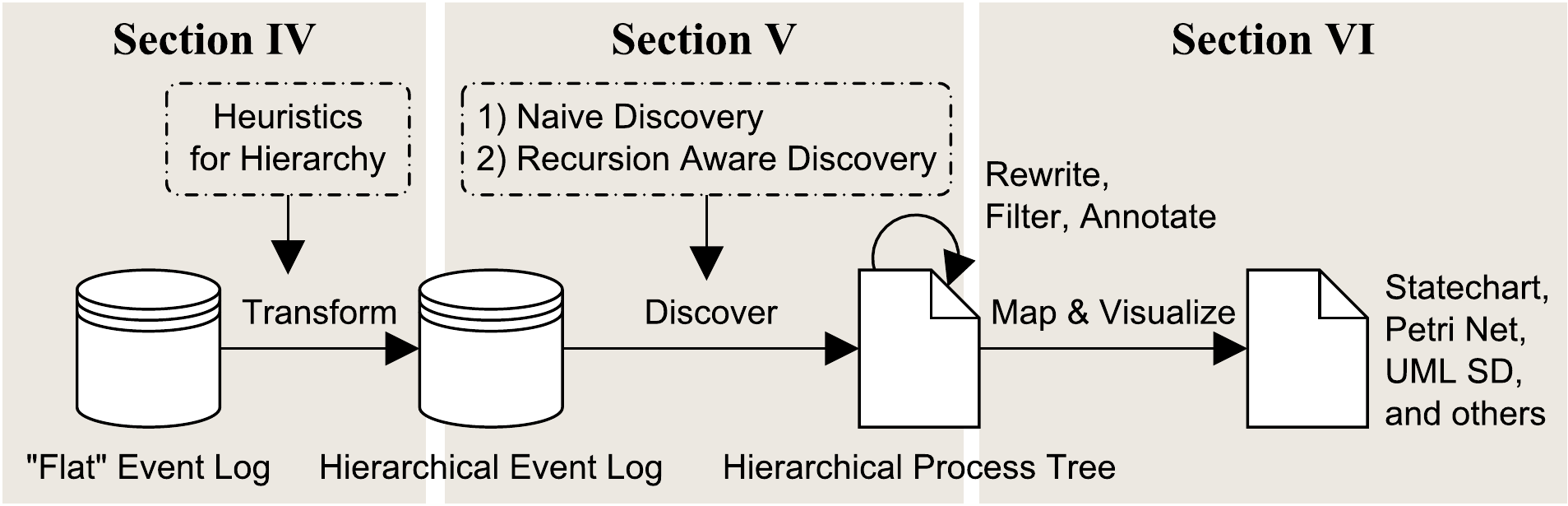}%
    \vspace*{-0.5em}
    \caption{Outline of the paper and our discovery approach.}%
    \label{fig:outline}%
\end{figure}%

% -- Content - Related Work 2 - RW - Intro --
\section{Related Work}
\label{sec:related-work}

Substantial work has been done on constructing models from software or example behavior in a variety of research domains.
%In this section, we make a brief comparison of a sample of the various approaches and focus mainly on the design criteria from the introduction.
This section presents a brief comparison of various approaches and focuses mainly on the design criteria from the introduction.
That is, the approach and tools should provide a behavioral backbone model that is expressive, precise and fits the actual system, and ideally should be able to support at least performance (timing) and frequency (usage) analysis.
%The comparison is summarized in Table~\ref{tab:related-work-compare}.
Table~\ref{tab:related-work-compare} summarizes the comparison.

% -- Content - Related Work 2 - RW - Groups and Criterea --
\subsection{Groups and Criteria for Comparison}

We have divided the related work into four groups.
\emph{Static Analysis} utilizes static artifacts like source code files.
\emph{Dynamic Analysis} utilizes runtime information through instrumentation or tracing interfaces like debuggers.
\emph{Grammar Inference} relies on example behavior in the form of abstract symbol sequences.
%/Process Mining/ relies on (dynamically obtained) events logs, and is in that sense a more implementation or platform agnostic approach.
\emph{Process Mining} relies on events logs and is in that sense a more implementation or platform agnostic approach.

For the comparison on design criteria, we define three sets of features.
Firstly, a precise and fit model should: a) have \emph{formal execution semantics}, and b) the underlying discovery algorithm should either guarantee or allow the user to control the \emph{model quality}.
The quality of a model is typically measured in terms of metrics like \emph{fitness} and \emph{precision}, but other qualities (e.g., simplicity and generalization) can also be considered~\cite{processmining}.
Fitness expresses the part of the log (c.q.
system behavior) that is represented by the model;
precision expresses the behavior in the model that is present in the log (c.q.
system behavior).
Secondly, the model should be used as the backbone for further analysis.
At the very least, \emph{frequency} (usage) and \emph{performance} (timing) analysis should be supported.
In addition, the analysis should be (statistically) significant, and hence the technique should be able to \emph{aggregate} information over multiple execution \emph{runs}.
Thirdly, the model should be expressive and be able to capture the type of behavior encountered in software system cases.
Not only should branching behavior like \emph{choices} (e.g., if-then-else) and \emph{loops} (e.g., foreach, iterators) be supported, but also hierarchy and recursion.
Furthermore, for hierarchies, meaningful \emph{names} for the different \emph{submodels} are also important.

% -- Content - Related Work 2 - RW - Discussion --
\subsection{Discussion of the Related Work}

In general, static and symbolic analysis of software has difficulty capturing the actual, dynamic behavior; especially in the case of dynamic types (e.g., inheritance, dynamic binding, exception handling) and jumps.
In these cases, it is often favorable to observe the actual system for the behavior.
Since static techniques either unfold or do not explore function calls, they lack support for recursive behavior.
In addition, because these techniques only look at static artifacts, they lack any form of timing or usage analysis.

In the area of dynamic analysis, the focus is on obtaining a rich but flat control flow model.
A lot of effort has been put in enriching models with more accurate choice and loop information, guards, and other predicates.
However, notions of recursion or preciseness of models, or application of these models, like for analysis, seems to be largely ignored.
The few approaches that do touch upon performance or frequency analysis (\cite{ackermann2009recovering,DePauw1998execpattern,Graham1982gprof}) do so with models lacking formal semantics or model quality guarantees.

In contrast to dynamic analysis techniques, grammar inference approaches are actively looking for repeating sub patterns (i.e., sources for hierarchies).
The used grammars have a strong formal basis.
However, in the grammar inference domain, abstract symbols are assumed as input, and the notion of branching behavior (e.g., loops) or analysis is lost.

In the area of process mining, numerous techniques have been proposed.
These techniques have strong roots in Petri nets, model conversions, and alignment-based analysis~\cite{adriansyah2014aligning,sleemans-thesis}
Process mining techniques yield formal models directly usable for advanced performance, frequency and conformance analysis.
There are only a few techniques in this domain that touch upon the subject of hierarchies.
In~\cite{gunther2007fuzzy,Bose2012HPM}, a hierarchy of anonymous clusters is created based on behavioral abstractions.
The hierarchy of anonymous clusters in~\cite{conforti2016bpmnminer} is based on functional and inclusion dependency discovery techniques over extra data in the event log.
None of these techniques yields named submodels or supports recursion.

Process mining techniques rely on event logs for their input.
These event logs can easily be obtained via the same techniques used by dynamic analysis for singular and distributed systems \cite{leemans2015models}.
Example techniques include, but are not limited to, Java Agents, Javassist \cite{chiba1998javassist, chiba200loadtime-javassist}, AspectJ \cite{gradecki2003mastering}, AspectC++ \cite{spinczyk2002aspectcpp}, AOP++ \cite{yao2005aoppp} and the TXL preprocessor \cite{cordy2006txl}.

% -- Content - Definitions --
\section{Definitions}
\label{sec:definitions}

Before we explain the proposed discovery techniques, we first introduce the definitions for our input and internal representation.
We start with some preliminaries in Subsection~\ref{sec:preliminaries}.
In Subsection~\ref{sec:eventlogs} we introduce two types of event logs: ``flat'' event logs (our input), and hierarchical event logs (used during discovery).
Finally, in Subsection~\ref{sec:process-trees} and~\ref{sec:hier-process-trees}, we will discuss the process tree model and our novel extension: \emph{hierarchical process trees}.
Throughout this paper, we will be using the program given in Listing~\ref{lst:running-example} as a running example, and assume we can log the start and end of each method.

\input{gfx/model/sw-code.tex}

% -- Content - Definitions - Preliminaries-1 --
\subsection{Preliminaries}
\label{sec:preliminaries}

\subsubsection{Multisets}
We denote the set of all multisets over some set $A$ as $\bag{A}$.
Note that the ordering of elements in a multiset is irrelevant.

\subsubsection{Sequences}
Given a set~$X$, a sequence over~$X$ of length~$n$ is denoted as~$\seq{a_1, \ldots, a_n} \in \setseq{X}$.
We define $\seq{a_1, \ldots, a_n}[i] = a_i$.
The empty sequence is denoted as~$\symEmptyTrace$.
Note that the ordering of elements in a sequence is relevant.
We write $\seqcat$ to denote sequence concatenation, for example: $\seq{a} \seqcat \seq{b} = \seq{a, b}$ , and $\seq{a} \seqcat \symEmptyTrace = \seq{a}$.
We write $\project{x}{i}$ to denote sequence (tail) projection, where $\project{\seq{a_1, a_2, \ldots, a_n}}{i} = \seq{a_i, \ldots, a_n}$.
For example: $\project{\seq{a, b}}{0} = \seq{a, b}$, $\project{\seq{a, b}}{1} = \seq{b}$, and $\project{\seq{a, b}}{2} = \symEmptyTrace$.
We write $\seqshuffle$ to denote sequence interleaving (shuffle).
For example: $\seq{a,b} \seqshuffle \seq{c, d} = \{\, \seq{a,b,c,d},$ $ \seq{a,c,b,d}, \seq{a,c,d,b}, \seq{c,a,b,d}, \seq{c,a,d,b}, \seq{c,d,a,b} \,\}$.

% -- Content - Definitions - Logs --
\subsection{Event Logs}
\label{sec:eventlogs}

% -- Content - Definitions - Logs - Flat Event Log-1 --
\subsubsection{``Flat'' Event Logs}

The starting point for any process mining technique is an \emph{event log}, a set of \emph{events} grouped into \emph{traces}, describing what happened when.
Each trace corresponds to an execution of a process; typically representing an example run in a software context.
Various \emph{attributes} may characterize events, e.g., an event may have a timestamp, correspond to an activity, denote a start or end, reference a line number, is executed by a particular process or thread, etc.

% -- Content - Definitions - Logs - Hierarchical Event Log --
\subsubsection{Hierarchical Event Logs}
\label{sec:def-hier-event-logs}

A hierarchical event log extends on a ``flat'' event log by assigning multiple activities to events; each activity describes what happened at a different level of granularity.
We assume a ``flat'' event log as input.
Based on this input, we will create a hierarchical event log for discovery.
For the sake of clarity, we will ignore most event attributes, and use sequences of activities directly, as defined below.

\begin{define}{Hierarchical Event Log}{def:Heventlog}
Let~$\actAlph$ be a set of activities.
Let~$\symEventLog \in \bag{\setseq{(\setseq{\actAlph})}}$ be a hierarchical event log, a multiset of traces.
A trace $t \in \symEventLog$, with $t \in \setseq{(\setseq{\actAlph})}$, is a sequence of events.
Each event $x \in t$, with $x \in \setseq{\actAlph}$, is described by a sequence of activities, stating which activity was executed at each level in the hierarchy.
\end{define}

Consider, for example, the hierarchical event log $\symEventLog = \mset{ \seq{ \seq{g,a}, \seq{g,b}, \seq{c} } }$.
This log has one trace, where the first event is labeled $\seq{g, a}$, the second event is labeled $\seq{g, b}$, and the third event is labeled $\seq{c}$.
For the sake of readability, we will use the following shorthand notation: $\symEventLog = \mset{ \seq{g.a, g.b, c} }$.
In this example log, we have two levels in our hierarchy: the longest event label has length~2, notation: $\|\symEventLog\| = 2$.
Complex behavior, like choices, loops and parallel (interleaved) behavior, is typically represented in an event log via multiple (sub)traces, showing the different execution paths.

We write the following to denote hierarchy concatenation: $f.\seq{g.a, g.b, c} = \seq{f.g.a, f.g.b, f.c}$.
We generalize concatenation to hierarchical logs: $f.\symEventLog = \msetc{ f.\symTrace }{ \symTrace \in \symEventLog}$.

We extend sequence projection to hierarchical traces and logs, such that a fixed length prefix is removed for all events: $\seqproject{\seq{g.a, g.b, c}}{0} = \seq{g.a, g.b, c}$, $\seqproject{\seq{g.a, g.b, c}}{1} = \seq{a, b}$, $\seqproject{\seq{g.a, g.b, c}}{2} = \symEmptyTrace$.
For logs: $\seqproject{\symEventLog}{i} = \msetc{ \seqproject{\symTrace}{i} }{ \symTrace \in \symEventLog }$.

In Table~\ref{tab:ex-htrace}, an example hierarchical trace is shown.
Here, we used the class plus method name as activities.
While generating logs, one could also include the full package name (i.e., a canonical name), method parameter signature (to distinguish overloaded methods), and more.

% -- Content - Definitions - Process Tree 2 --
\subsection{Process Trees}
\label{sec:process-trees}

% -- Content - Definitions - Process Tree 2 - Process Tree Intro --
In this subsection, we introduce \emph{process trees} as a notation to compactly represent \emph{block-structured models}.
An important property of block-structured models is that they are \emph{sound by construction}; they do not suffer from deadlocks, livelocks, and other anomalies.
In addition, process trees are tailored towards process discovery and have been used previously to discover block-structured workflow nets~\cite{sleemans-thesis}.
A process tree describes a language; an operator describes how the languages of its subtrees are to be combined.

% -- Content - Definitions - Process Tree 2 - Process Tree Model --
\begin{define}{Process Tree}{def:processtree}
We formally define \emph{process trees} recursively.
We assume a finite alphabet $\actAlph$ of activities and a set $\treeSetOp$ of operators to be given.
Symbol $\actSilent \notin \actAlph$ denotes the silent activity.
\begin{itemize}
    \item $a$ with $a \in (\actAlph \cup \set{\actSilent})$ is a process tree;
    \item Let $\symTree_1, \ldots, \symTree_n$ with $n > 0$ be process trees and let $\treeOp \in \treeSetOp$ be a process tree operator, then $\treeOp(\symTree_1, \ldots, \symTree_n)$ is a process tree.
\end{itemize}

We consider the following operators for process trees:
\noindent
\begin{tabularx}{0.5\textwidth}{>{\centering\arraybackslash} X p{0.45\textwidth}}
    $\treeOpSeq$    & denotes the \emph{sequential execution} of all subtrees \\
    $\treeOpXor$    & denotes the \emph{exclusive choice} between one of the subtrees \\
    $\treeOpLoop$   & denotes the \emph{structured loop} of loop body $\symTree_1$ and alternative loop back paths $\symTree_2, \ldots, \symTree_n$ (with $n \geq 2$) \\
    $\treeOpPar$    & denotes the \emph{parallel (interleaved) execution} of all subtrees \\
\end{tabularx}
\vspace*{-0.8em}
\end{define}

% -- Content - Definitions - Process Tree 2 - Process Tree Language --
To describe the semantics of process trees, the language of a process tree $\symTree$ is defined using a recursive monotonic function $\treeLan(\symTree)$, where each operator $\treeOp$ has a language join function $\treeOp^l$:

{\vspace*{-2ex}\small\begin{align*}
	\treeLan(a) 	&\,= \set{\seq{a}} \text{\textbf{ for }} a \in \actAlph \\
	\treeLan(\actSilent)	&\,= \set{\symEmptyTrace} \\
\treeLan(\treeOp(\symTree_1, \ldots, \symTree_n)) &\,= \treeOp^l(\treeLan(\symTree_1), \ldots, \treeLan(\symTree_n))%
\end{align*}\vspace*{-2ex}}

Each operator has its own language join function $\treeOp^l$.
The language join functions below are borrowed from~\cite{sleemans-thesis, processmining}:

{\vspace*{-2ex}\footnotesize\begin{align*}
  \treeOpSeq^l(L_1, \ldots, L_n) =\,& \setc{
    t_1 \seqcat \ldots \seqcat t_n
  }{
    \forall 1 \leq i \leq n : t_i \in L_i
  } \\
  \treeOpXor^l(L_1, \ldots, L_n) =\,& \textstyle \bigcup_{1 \leq i \leq n} L_i \\
  \treeOpLoop^l(L_1, \ldots, L_n) =\,& \{\,
    t_1 \seqcat t'_1 \seqcat t_2 \seqcat t'_2 \seqcat \ldots \seqcat t_{m-1} \seqcat t'_{m-1} \seqcat t_m \\
  &\;|\; \textstyle
    \forall i : t_i \in L_1 ,\, t'_i  \in \bigcup_{2 \leq j \leq n}
    L_j \,\} \\
  \treeOpPar^l(L_1, \ldots, L_n) =\,& \setc{
    t' \in (t_1 \seqshuffle \ldots \seqshuffle t_n)
  }{
    \forall 1 \leq i \leq n : t_i \in L_i
  }%
\end{align*}}% 

Example process trees and their languages:

{\vspace*{-2ex}\small\begin{align*}
  \treeLan(\treeOpSeq(a, \treeOpXor(b, c))) =\,& \set{ \seq{a, b}, \seq{a, c} } \\
  \treeLan(\treeOpPar(a, b)) =\,& \set{ \seq{a, b}, \seq{b, a} } \\
  \treeLan(\treeOpPar(a, \treeOpSeq(b, c)) =\,& \set{ \seq{a, b, c},  \seq{b, a, c}, \seq{b, c, a} } \\
  \treeLan(\treeOpLoop(a, b)) =\,& \set{ \seq{a}, \seq{a, b, a}, \seq{a,b,a,b,a}, \ldots }%
\end{align*}}%

% -- Content - Definitions - Hierarchy PT --
\subsection{Hierarchical Process Trees}
\label{sec:hier-process-trees}

% -- Content - Definitions - Hierarchy PT - HPT Model --
We extend the process tree representation to support hierarchical and recursive behavior.
We add a new tree operator to represent a named submodel, and add a new tree leaf to denote a recursive reference.
Figure~\ref{fig:ex-ptree} shows an example model.

\begin{define}{Hierarchical Process Tree}{def:hier-process-tree}
We formally define \emph{hierarchical process trees} recursively.
We assume a finite alphabet $\actAlph$ of activities to be given.
\begin{itemize}
    \item Any process tree is also a hierarchical process tree;
    \item Let $\symTree$ be a hierarchical process tree, and~$f \in \actAlph$, then $\treeOpCompOr_f(\symTree)$ is a hierarchical process tree that denotes the named subtree, with name~$f$ and subtree $\symTree$;
    \item $\treeOpRecurOr_f$ with $f \in \actAlph$ is a hierarchical process tree.
Combined with a named subtree operator $\treeOpCompOr_f$, this leaf denotes the point where we recurse on the named subtree.
See also the language definition below.
\end{itemize}
\vspace*{-0.8em}
\end{define}

% -- Content - Definitions - Hierarchy PT - HPT Language --
The semantics of hierarchical process trees are defined by extending the language function $\treeLan(\symTree)$.
A recursion at a leaf~$\treeOpRecurOr_f$ is `marked' in the language, and resolved at the level of the first corresponding named submodel~$\treeOpCompOr_f$.
Function~$\psi^l_f$ `scans' a single symbol and resolves for marked recursive calls (see $\fun{\psi^l_f}{ \treeOpRecurOr_f }$).
Note that via~$\psi^l_f$, the language $\treeOpCompOr^{l}_{f}$ is defined recursively and yields a hierarchical, recursive language.

{\vspace*{-2ex}\small\begin{align*}
  \treeLan(\treeOpRecurOr_f)
    =\,& \set{\seq{\treeOpRecurOr_f}} \text{\textbf{ for }} f \in \actAlph \\
  \treeOpCompOr^{l}_{f}(L)
    =\,& \{\, f .
( t'_1 \seqcat \ldots \seqcat t'_n ) \;|\; \seq{x_1, \ldots, x_n} \in L , \\
       &\;\; \forall 1 \leq i \leq n : t'_i \in \fun{\psi^l_f}{x_i} \,\} \text{\textbf{ for }} f \in \actAlph \\
  \text{\textbf{ where }}
    \fun{\psi^l_f}{ \symEmptyTrace }    =\,& \set{ \symEmptyTrace } \\
    \fun{\psi^l_f}{ \treeOpRecurOr_f }  =\,& \treeOpCompOr^{l}_{f}(L) \text{\textbf{ for }} f \in \actAlph \\
    \fun{\psi^l_f}{ \treeOpRecurOr_g }  =\,& \set{ \seq{\treeOpRecurOr_g} } \text{\textbf{ for }} g \neq f \land g \in \actAlph \\
    \fun{\psi^l_f}{ a .
x }             =\,& \setc{ a .
( t' ) }{ t' \in \fun{\psi^l_f}{ x } } \text{\textbf{ for }} a \in \actAlph, x \in \setseq{\actAlph}
\end{align*}}

Example hierarchical process trees and their languages:

{\vspace*{-2ex}\small\begin{align*}
  \treeLan(\treeOpCompOr_f(\treeOpSeq(a, b)) =\,& \set{ \seq{f.a, f.b} } \\
  \treeLan(\treeOpCompOr_f(\treeOpSeq(a, \treeOpCompOr_g(b))) =\,& \set{ \seq{f.a, f.g.b} } \\
  \treeLan(\treeOpCompOr_f(\treeOpXor(\treeOpSeq(a, \treeOpRecurOr_f), b)) =\,& \{\, \seq{f.b}, \seq{f.a, f.f.b}, \\
&\;\, \seq{f.a, f.f.a, f.f.f.b}, \ldots \,\} \\
  \treeLan(\treeOpCompOr_f(\treeOpCompOr_g(\treeOpXor(a, \treeOpRecurOr_f))) =\,&
    \{\, \seq{f.g.a}, \seq{f.g.f.g.a}, \\
&\;\, \seq{f.g.f.g.f.g.a}, \ldots \,\} \\
  \treeLan(\treeOpCompOr_f(\treeOpCompOr_g(\treeOpXor(a, \treeOpRecurOr_f, \treeOpRecurOr_g))) =\,&
    \{\, \seq{f.g.a}, \seq{f.g.g.a}, \seq{f.g.f.g.a}, \\
&\;\, \seq{f.g.f.g.g.a}, \seq{f.g.g.f.g.a}, \ldots \,\}
\end{align*}} 

% -- Content - Definitions - Hierarchy PT - HPT Example --
\begin{figure}[!htb]%
    \centering%
    \vspace*{-1.5em}
    \input{gfx/model/sw-ptree.tex}%
    \caption{The hierarchical process tree model corresponding to the hierarchical event log with the trace from Table~\ref{tab:ex-htrace}.
Note that we modeled the recursion at \cf{B.process()} explicitly.}%
    \label{fig:ex-ptree}%
    \vspace*{-2.2ex}
\end{figure}
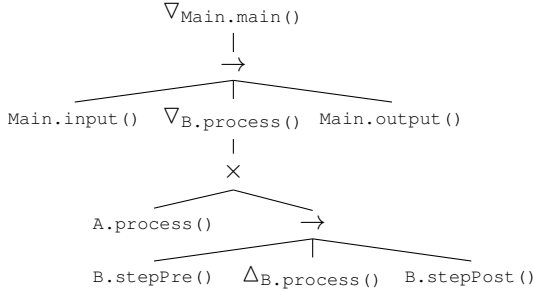% 

% -- Content - Heuristics for Hierarchy --
\section{Heuristics for Hierarchy}
\label{sec:heuristics}

To go from ``flat'' event logs (our input) to hierarchical event logs (used during discovery), we rely on transformation heuristics.
In the typical software system case, we will use the \emph{Nested Calls} heuristic, but our implementation also provides other heuristics for different analysis scenarios.
We will discuss three of these heuristics for hierarchy next.

1) \emph{Nested Calls} captures the ``executions call stacks''.
Through the use of the life-cycle attribute (start-end), we can view flat event logs as a collection of intervals, and use \emph{interval containment} to build our hierarchy.
In Figure~\ref{fig:ex-log-int}, a trace from the event log corresponding to the program in Listing~\ref{lst:running-example} is depicted as intervals.
Table~\ref{tab:ex-htrace} shows the corresponding ``nested calls'' hierarchical event log trace.
2) \emph{Structured Names} captures the static structure or ``architecture'' of the software source code.
By splitting activities like ``\cf{package.class.method()}'' on ``\cf{.}'' into\ $\seq{\text{\cf{package}},\, \text{\cf{class}},\, \text{\cf{method()}}}$, we can utilize the designed hierarchy in the code as a basis for discovery.
3) \emph{Attribute Combination} is a more generic approach, where we can combine several attributes associated with events in a certain order.
For example, some events are annotated with a high-level and a low-level description, like an interface or protocol name plus a specific function or message/signal name.

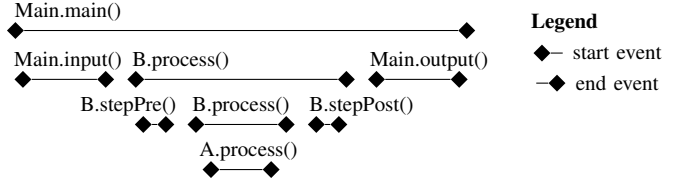
\begin{figure}[!htb]%
    \centering%
    \input{gfx/model/sw-log-int.tex}%
    \vspace*{-3ex}
    \caption{A trace from the event log corresponding to logging the methods in the program in Listing~\ref{lst:running-example}, depicted as intervals.
By using the \emph{Nested Call} heuristics, a hierarchical event log can be constructed from contained intervals, see also Table~\ref{tab:ex-htrace}.
For example, B.process() is contained in Main.main().}%
    \label{fig:ex-log-int}%
\end{figure}% 

\input{gfx/model/sw-log-2.tex}

% -- Content - Model Discovery --
\section{Model Discovery}
\label{sec:discovery}

% -- Content - Model Discovery - Discovery Framework --
\subsection{Discovery Framework}

Our techniques are based on the Inductive miner framework for discovering process tree models.
This framework is described in~\cite{sleemans-thesis}.
Given a set~$\treeSetOp$ of process tree operators, \cite{sleemans-thesis} defines a framework to discover models using a divide and conquer approach.
Given a log $\symEventLog$, the framework searches for possible splits of~$\symEventLog$ into sublogs~$\symEventLog_1, \ldots, \symEventLog_n$, such that these logs combined with an operator~$\treeOp$ can (at least) reproduce~$L$ again.
It then recurses on the corresponding sublogs and returns the discovered submodels.
Logs with empty traces or traces with a single activity form the base cases for this framework.
Note that the produced model can be a generalization of the log; see for example the language of the structured loop ($\treeOpLoop$).
In Table~\ref{tab:disc-exb-im}, an example run of the framework is given.

\input{./gfx/disc/exb-im-table.tex}

% -- Content - Model Discovery - Guarantees Include --
We present two adaptations of the framework described above: one for hierarchy (\emph{Na\"ive Discovery}, Subsection~\ref{sec:disc:naive}), and one for recursions (\emph{Recursion Aware Discovery}, Subsection~\ref{sec:disc:rad}).
Our adaptations maintain the termination guarantee, perfect fitness guarantee, language rediscoverability guarantee, and the polynomial runtime complexity from~\cite{sleemans-thesis}.
The details of the above guarantees and runtime complexity for our adaptations are detailed in Subsection~\ref{sec:guarantees}.
In our implementation, we rely on the (in)frequency based variant to enable the discovery of an 80/20 model, see also Subsection~\ref{sec:infreqmodel}.

% -- Content - Model Discovery - Naive Discovery --
\subsection{Naive Discovery}
\label{sec:disc:naive}

% -- Content - Model Discovery - Naive Discovery - ND - Definition --
We generalize the above approach to also support hierarchy by providing the option to split hierarchical event logs and use hierarchical sequence projection.

Using our generalized discovery framework, we define a naive realization associated with the named subtree operator~$\treeOpCompOr_f$, and a slightly modified base case.
The details are given in Algorithm~\ref{alg:naive-discovery}.
In Table~\ref{tab:disc-ex3-nd}, an example run is given.

This algorithmic variant is best suited for cases where recursion does not make much sense, for example, when we are using a hierarchical event log based on structured package-class-method name hierarchies.

\input{./gfx/disc/alg-naive.tex}
Below are some more example logs and the models discovered.
Note that with this naive approach, the recursion in the last example is not discovered.
{\small\begin{align*}
       \funn{Naive}(\mset{\seq{f.a, f.b}, \seq{f.c}})
  =\,& \treeOpCompOr_f(\funn{Naive}(\mset{\seq{a, b}, \seq{c}})) \\
%  =\,& \treeOpCompOr_f(\treeOpXor(\funn{Naive}(\mset{\seq{a, b}}), c)) \\
  =\,& \treeOpCompOr_f(\treeOpXor(\treeOpSeq(a, b), c)) \\
       \funn{Naive}(\mset{\seq{f.a, f.g.f.b}})
  =\,& \treeOpCompOr_f(\funn{Naive}(\mset{\seq{a, g.f.b}})) \\
  =\,& \treeOpCompOr_f(\treeOpSeq(a, \funn{Naive}(\mset{\seq{g.f.b}}) )) \\
%  =\,& \treeOpCompOr_f(\treeOpSeq(a, \treeOpCompOr_g( \funn{Naive}(\mset{\seq{f.b}}) ))) \\
  =\,& \treeOpCompOr_f(\treeOpSeq(a, \treeOpCompOr_g( \treeOpCompOr_f(b) ))) \\
       \funn{Naive}(\mset{\seq{f.a}, \seq{f}})
  =\,& \treeOpCompOr_f(\funn{Naive}(\mset{\seq{a}, \symEmptyTrace})) \\
  =\,& \treeOpCompOr_f(\treeOpXor(a, \tau))
\end{align*}}%

% -- Content - Model Discovery - Recursion Aware Discovery --
\subsection{Recursion Aware Discovery}
\label{sec:disc:rad}

% -- Content - Model Discovery - Recursion Aware Discovery - RAD - Idea --
In order to successfully detect recursion, we make some subtle but important changes.
We rely on two key notions: 1)~a~\emph{context path}, and 2)~\emph{delayed discovery}.
Both are explained below, using the example shown in Table~\ref{tab:disc-ex3-rd}.

This algorithmic variant is best suited for cases where recursion makes sense, for example, when we are using an event log based on the \emph{Nested Calls} hierarchy (Section~\ref{sec:heuristics}).

To detect recursion, we need to keep track of the named subtrees from the root to the current subtree.
We call the sequence of activities on such a path the \emph{context path}, notation~$\symContextPath \in \setseq{\actAlph}$.
The idea is that whenever we discover a named subtree $\treeOpCompOr_f$, and we encounter another activity $f$ somewhere in the sublogs, we can verify this recursion using~$f \in \symContextPath$.
Sublogs collected during discovery are associated with a context path, notation~$\symEventLog(\symContextPath)$.
This approach is able to deal with complex recursions (see examples at the end) and overloaded methods (see the activity naming discussion in Section~\ref{sec:def-hier-event-logs}).
In Table~\ref{tab:disc-ex3-rd}, the current context path~$\symContextPath$ at each step is shown.

\vspace*{-1ex}%
\input{./gfx/disc/ex3-nd-table.tex}
\vspace*{-3ex}%
\input{./gfx/disc/ex3-rd-table.tex}

Let's have a closer look at steps~2 and~3 in Table~\ref{tab:disc-ex3-rd}.
Note how the same subtree is discovered twice.
In step~2, we detect the recursion.
And in step~3, we use the sublog after the recursion part as an additional set of traces.
The idea illustrated here is that of \emph{delayed discovery}.
Instead of immediately discovering the subtree for a named subtree $\treeOpCompOr_f$, we delay that discovery.
The corresponding sublog is associated with the current context path.
For each context path, we discover a model for the associated sublog.
During this discovery, the sublog associated with that context path may change.
If that happens, we run that discovery again on the extended sublog.
Afterwards, we insert the partial models under the corresponding named subtrees operators.

Algorithm~\ref{alg:recurse-discovery} details the recursion aware discovery algorithm;
it uses Algorithm~\ref{alg:recurse-discovery-run} for a single discovery run.
In the example of Table~\ref{tab:disc-ex3-rd}, we first discover on the complete log with the empty context path (Alg.~\ref{alg:recurse-discovery}, line~\ref{alg:recurse-discovery:root}).
In step~1, we encounter the named subtree $\treeOpCompOr_f$, and associate
$\symEventLog(\seq{f}) = \mset{\seq{a, f.b}}$, for context path $\symContextPath = \seq{f}$ 
(Alg.~\ref{alg:recurse-discovery-run}, line~\ref{alg:recurse-discovery-run:subtree}).
In step~2, we start discovery on $\symContextPath = \seq{f}$ using the sublog $\symEventLog(\seq{f})$ (Alg.~\ref{alg:recurse-discovery}, line~\ref{alg:recurse-discovery:iterate}).
In this discovery, we encounter the recursion $f \in \symContextPath$, and add $\mset{\seq{b}}$ to the sublog, 
resulting in $\symEventLog(\seq{f}) = \mset{\seq{a, f.b}, \seq{b}}$ (Alg.~\ref{alg:recurse-discovery-run}, line~\ref{alg:recurse-discovery-run:recursion}).
Finally, in step~3, we rediscover for $\symContextPath = \seq{f}$, now using the extended sublog (Alg.~\ref{alg:recurse-discovery}, line~\ref{alg:recurse-discovery:iterate}).
In this discovery run, no sublog changes anymore.
We insert the partial models under the corresponding named subtrees operators
(Alg.~\ref{alg:recurse-discovery}, line~\ref{alg:recurse-discovery:glue}) and return the result.

\input{./gfx/disc/alg-recurse-rad.tex}
\vspace*{-1ex}%
\input{./gfx/disc/alg-recurse-radrun.tex}
Below are some more example logs, the models discovered, and the sublogs associated with the involved context paths.
Note that with this approach, complex recursions are also discovered.
{\small
\begin{align*}
    \funn{RAD}(\mset{\seq{f.a, f.g.f.b}}) 
    =\,&
    \treeOpCompOr_f(\treeOpXor(b, \treeOpSeq(a, \treeOpCompOr_g( \treeOpRecurOr_f )))) \\
    \text{\textbf{where }} \symEventLog(\seq{f}) =\,& \mset{\seq{b}, \seq{a, g.f.b}} \\
                           \symEventLog(\seq{f,g}) =\,& \mset{\seq{f.b}} \\
    \funn{RAD}(\mset{\seq{f.g.g.a}, \seq{f.g.f.g.a}}) 
    =\,& 
    \treeOpCompOr_f(\treeOpCompOr_g(\treeOpXor(a, \treeOpRecurOr_f, \treeOpRecurOr_g))) \\
    \text{\textbf{where }} \symEventLog(\seq{f}) =\,& \mset{\seq{g.g.a}, \seq{g.f.g.a}, \seq{g.a}} \\
                           \symEventLog(\seq{f,g}) =\,& \mset{\seq{g.a}, \seq{f.g.a}, \seq{a}} \\
    \funn{RAD}(\mset{\seq{f.f}}) 
    =\,& 
    \treeOpCompOr_f(\treeOpXor(\treeOpRecurOr_f, \tau)) \\
    \text{\textbf{where }} \symEventLog(\seq{f}) =\,& \mset{\seq{f}, \symEmptyTrace}
\end{align*}}%

% -- Content - Model Discovery - Guarantees Detail --
\subsection{Termination, Perfect Fitness, Language Rediscoverability, and Runtime Complexity}
\label{sec:guarantees}

Our Na\"ive Discovery and Recursion Aware Discovery adaptations of the  framework described \cite{sleemans-thesis} maintains the termination guarantee, perfect fitness guarantee, language rediscoverability guarantee, and the polynomial runtime complexity.
We will discuss each of these properties using the simplified theorems and proofs from~\cite{sleemans2013-im}.

% -- Content - Model Discovery - Guarantees Detail - Termination --
\subsubsection{Termination Guarantee}

The termination guarantee is based on the proof for~\cite[Theorem~2, Page~7]{sleemans2013-im}.
The basis for the termination proof relies on the fact that the algorithm only performs finitely many recursions.
For the standard process tree operators in the original framework, it is shown that the log split operator only yields finitely many sublogs.
Hence, for our adaptations, we only have to show that the new hierarchy and recursion cases only yield finitely many recursions.

\begin{theorem}
Na\"ive Discovery terminates.
\end{theorem}
\begin{proof}
Consider the named subtree case on Algorithm~\ref{alg:naive-discovery}, line~\ref{alg:naive-discovery:subtree}.
Observe that the log $\symEventLog$ has a finite depth, i.e., a finite number of levels in the hierarchy.
Note that the sequence projection~$\seqproject{\symEventLog}{1}$ yields 
strictly smaller event logs, i.e, the number of levels in the hierarchy strictly decreases.
We can conclude that the named subtree case for the Na\"ive Discovery yields only finitely many recursions.
Hence, the Na\"ive Discovery adaptation maintains the termination guarantee of~\cite[Theorem~2, Page~7]{sleemans2013-im}.
\end{proof}

\begin{theorem}
Recursion Aware Discovery terminates.
\end{theorem}
\begin{proof}
Consider the named subtree and recursion cases in Algorithm~\ref{alg:recurse-discovery-run} on
lines~\ref{alg:recurse-discovery-run:recursion} and~\ref{alg:recurse-discovery-run:subtree}.
Note that, by construction, for all the cases where we end up in Algorithm~\ref{alg:recurse-discovery-run}, line~\ref{alg:recurse-discovery-run:recursion}, 
we know that $\symEventLog$ is derived from, and bounded by, $\symEventLog(\symContextPath')$ as follows: 
$\symEventLog \subseteq \setc{ \seqproject{\symEventLog'}{i} }{ \symEventLog' \subseteq \symEventLog(\symContextPath') \land 0 \leq i \leq \|\symEventLog(\symContextPath')\| }$.
Observe that the log $\symEventLog$ has a finite depth, i.e., a finite number of levels in the hierarchy.
Note that the sequence projection~$\seqproject{\symEventLog}{1}$ yields strictly smaller event logs, 
i.e, the number of levels in the hierarchy strictly decreases.
Hence, we can conclude that $\symEventLog(\symContextPath')$ only changes finitely often.
Since $\symContextPath$ is derived from the log depth, we also have a finitely many sublogs 
$\symEventLog(\symContextPath')$ that are being used.
Hence, the loop on Algorithm~\ref{alg:recurse-discovery}, line~\ref{alg:recurse-discovery:iterate} terminates, 
and thus the Recursion Aware Discovery adaptation maintains the termination guarantee of~\cite[Theorem~2, Page~7]{sleemans2013-im}.
\end{proof}

% -- Content - Model Discovery - Guarantees Detail - Perfect Fitness --
\subsubsection{Perfect Fitness}

As stated in the introduction, we want the discovered model to fit the actual behavior.
That is, we want the discovered model to at least contain all the behavior in the event log.
The perfect fitness guarantee states that all the log behavior is in the discovered model, 
and we proof this using the proof for~\cite[Theorem~3, Page~7]{sleemans2013-im}.
The fitness proof is based on induction on the log size\footnote{Formally, the original induction is on the log size plus a counter parameter.
However, for our proofs, we can ignore this counting parameter.}.
As induction hypothesis, we assume that for all sublogs, the discovery framework returns a fitting model, and then prove that the step maintains this property.
That is, for all sublogs~$\symEventLog'$ we have a corresponding submodel $\symTree'$ such that $\symEventLog' \subseteq \treeLan(\symTree')$.
For our adaptations, it suffices to show that the named subtree and recursion operators do not violate this assumption.

\begin{theorem}
Na\"ive Discovery returns a process model that fits the log.
\end{theorem}
\begin{proof}
By simple code inspection on Algorithm~\ref{alg:naive-discovery}, line~\ref{alg:naive-discovery:subtree} and using the induction hypothesis on~$\seqproject{\symEventLog}{1}$, we can see that for the named subtree operator we return a process model that fits the log~$\symEventLog$.
Since this line is the only adaptation, the Na\"ive Discovery adaptation maintains the perfect fitness guarantee of~\cite[Theorem~3, Page~7]{sleemans2013-im}.
\end{proof}

\begin{theorem}
Recursion Aware Discovery returns a process model that fits the log.
\end{theorem}
\begin{proof}
Consider the named subtree case on 
Algorithm~\ref{alg:recurse-discovery-run}, line~\ref{alg:recurse-discovery-run:subtree}.
Using the induction hypothesis on~$\symEventLog(\symContextPath \seqcat \seq{f}) = \seqproject{\symEventLog}{1}$, we know that $model(\symContextPath \seqcat \seq{f})$ will fit $\symEventLog(\symContextPath \seqcat \seq{f})$.
By Algorithm~\ref{alg:recurse-discovery}, line~\ref{alg:recurse-discovery:glue}, we know that $\funn{model}(\symContextPath \seqcat \seq{f})$ will be the child of~$\treeOpCompOr_f$.
Hence, for the named subtree operator we return a process model that fits the log~$\symEventLog$.

Consider the recursion case on 
Algorithm~\ref{alg:recurse-discovery-run}, line~\ref{alg:recurse-discovery-run:recursion}.
Since $f \in \symContextPath$, we know there must exist a named subtree~$\treeOpCompOr_f$ corresponding to the recursive operator~$\treeOpRecurOr_f$.
Due to Algorithm~\ref{alg:recurse-discovery}, line~\ref{alg:recurse-discovery:iterate} and the induction hypthesis, 
we know that at the end $\funn{model}(\symContextPath')$ fits $\symEventLog(\symContextPath')$ (i.e., $\symEventLog(\symContextPath') \subseteq \treeLan(\funn{model}(\symContextPath'))$).
Since, by construction, we know $\seqproject{\symEventLog}{1} \subseteq \symEventLog(\symContextPath')$,
$\funn{model}(\symContextPath')$ also fits $\seqproject{\symEventLog}{1}$.
By Algorithm~\ref{alg:recurse-discovery}, line~\ref{alg:recurse-discovery:glue}, we know that $\treeOpRecurOr_f$ will be in the subtree of~$\treeOpCompOr_f$.
Hence, for the recursion operator we return a process model that fits the log~$\symEventLog$.

We conclude that the Recursion Aware Discovery adaptation maintains the perfect fitness guarantee of~\cite[Theorem~3, Page~7]{sleemans2013-im}.
\end{proof}

% -- Content - Model Discovery - Guarantees Detail - Language Rediscoverability --
\subsubsection{Language Rediscoverability}

The language rediscoverability property tells whether and under which conditions
a discovery algorithm can discover a model that is language-equivalent to the original process.
That is, given a `system model'~$\symTree$ and an event log~$\symEventLog$ that
is complete w.r.t.
$\symTree$ (for some notion of completeness),
then we rediscover a model~$\symTree'$ such that $\treeLan(\symTree') = \treeLan(\symTree)$.

We will show language rediscoverability in several steps.
First, we will define the notion of language complete logs.
Then, we define the class of models that can be language-rediscovered.
And finally, we will detail the language rediscoverability proofs.

\paragraph{Language Completeness}
Language rediscoverability holds for directly-follows complete logs.
We adapt this notion of directly-folllows completeness from~\cite{sleemans2013-im}
by simply applying the existing definition to hierarchical event logs:

\newcommand{\opDFComplete}{\mathbin{\diamond_{df}}}

\begin{definition}[Directly-follows completeness]
Let $\fun{Start}{\symEventLog}$ and $\fun{End}{\symEventLog}$ denote the set of 
start and end symbols amongst all traces, respectively.
A log $\symEventLog$ is directly-follows complete to a model $\symTree$, 
denoted as $\symEventLog \opDFComplete \symTree$, iff:
\begin{enumerate}
  \item $\seq{\ldots, x, y, \ldots} \in \treeLan(\symTree) \limpl \seq{\ldots, x, y, \ldots} \in \symEventLog$;
  \item $\fun{Start}{\treeLan(\symTree)} \subseteq \fun{Start}{\symEventLog}$;
  \item $\fun{End}{\treeLan(\symTree)} \subseteq \fun{End}{\symEventLog}$; and
  \item $\Sigma(\symTree) \subseteq \Sigma(\symEventLog)$.

\end{enumerate}
Note that directly-follows completeness is defined over all levels of a hierarchical log.
\end{definition}

\paragraph{Class of Language-Rediscoverable Models}
We will prove language rediscoverability for the following class of models.
Let~$\Sigma(\symTree)$ denote the set of activities in~$\symTree$.
A model~$\symTree$ is in the class of language rediscoverable models iff
for all nodes $\treeOp(\symTree_1, \ldots, \symTree_n)$ in $\symTree$ we have:
\begin{enumerate}
  \item No duplicate activities: 
  $\forall i \neq j : \Sigma(\symTree_i) \cap \Sigma(\symTree_j) = \emptyset$;
  \item In the case of a loop, the sets of start and end activities of the first branch must be disjoint: \newline
  $\treeOp = \treeOpLoop \limpl \funn{Start}(\treeLan(\symTree_1)) \cap \funn{End}(\treeLan(\symTree_1)) = \emptyset$
  \item No taus are allowed: $\forall i \leq n : \symTree_i \neq \tau$;
  \item In the case of a recursion node~$\treeOpRecurOr_f$, 
  there exists a corresponding named subtree node~$\treeOpCompOr_f$
  on the path from~$\symTree$ to~$\treeOpRecurOr_f$.
\end{enumerate}
Note that the first three criteria follow directly from the language rediscoverability class
from~\cite{sleemans2013-im}.
The last criteria is added to have well-defined recursions in our hierarchical process trees.

\paragraph{Language-Rediscoverable Guarantee}

The language rediscoverability guarantee is based on the proof for~\cite[Theorem~14, Page~16]{sleemans2013-im}.
The proof in~\cite{sleemans2013-im} is based on three lemmas:
\begin{itemize}
  \item \cite[Lemma~11, Page~15]{sleemans2013-im}  guarantees that any root process tree operator is rediscovered;
  \item \cite[Lemma~12, Page~16]{sleemans2013-im} guarantees that the base cases can be rediscovered; and
  \item \cite[Lemma~13, Page~16]{sleemans2013-im} guarantees that for all process tree operators the log is correctly subdivided.
\end{itemize}

For our adaptations, we have to show:
\begin{enumerate}
  \item Our recursion base case maintains \cite[Lemma~12, Page~16]{sleemans2013-im}; and
  \item Our named subtree operator maintains \cite[Lemma~11, Page~15]{sleemans2013-im} 
and \cite[Lemma~13, Page~16]{sleemans2013-im}.
\end{enumerate}

\begin{theorem}\label{thm:lang-redisc:naive}
Na\"ive Discovery preserves language rediscoverability.
\end{theorem}
\begin{proof}
We only have to show that the introduction of the named subtree operator 
maintains language rediscoverability.

First, we show for the named subtree operator that the  root process tree operator is rediscovered (Lemma~11).
Assume a process tree~$\symTree = \treeOpCompOr_f(\symTree_1)$, for any $f \in \actAlph$,
and let $\symEventLog$ be a log such that $\symEventLog \opDFComplete \symTree$.
Since we know that $\symEventLog \opDFComplete \symTree$, we know that
$\forall x \in t \in \symEventLog : x[1] = f$, and there must be a lower level in the tree.
By simple code inspection on Algorithm~\ref{alg:naive-discovery}, line~\ref{alg:naive-discovery:subtree},
we can see that the Na\"ive Discovery will yield $\treeOpCompOr_f$.

Next, we show for the named subtree operator that the log is correctly subdivided (Lemma~13).
That is, lets assume:
1) a model $\symTree = \treeOpCompOr_f(\symTree_1)$ adhering to the model restrictions; and 
2) $\symEventLog \subseteq \treeLan(\symTree) \land \symEventLog \opDFComplete \symTree$.
Then we have to show that any sublog~$\symEventLog_i$ we recurse upon has:
$\symEventLog_i \subseteq \treeLan(\symTree_i) \land \symEventLog_i \opDFComplete \symTree_i$.
For the named subtree operator, we have exactly one sublog we recurse upon:
$\symEventLog_1 = \seqproject{\symEventLog}{1}$.
We can easily prove this using the sequence projection on the inducation hypothesis:
$\seqproject{\symEventLog}{1} \subseteq \seqproject{\treeLan(\symTree)}{1}$, after substitution:
$\seqproject{\symEventLog}{1} \subseteq \seqproject{\treeLan(\treeOpCompOr_f(\symTree_1))}{1}$.
By definition of the semantics for~$\treeOpCompOr_f$, we can rewrite this to:
$\seqproject{\symEventLog}{1} \subseteq \treeLan(\symTree_1)$.
The proof construction for  $\symEventLog_i \opDFComplete \symTree_i$ is analogous.
Hence, for the named subtree operator that the log is correctly subdivided.

We can conclude that the Na\"ive Discovery adaptation preserves language rediscoverability guarantee of~\cite[Theorem~14, Page~16]{sleemans2013-im}.
\end{proof}

\begin{theorem}
Recursion Aware Discovery preserves language rediscoverability.
\end{theorem}
\begin{proof}
The proof for the introduction of the named subtree operator is analogous to
the proof for Theorem~\ref{thm:lang-redisc:naive},
using the fact that always $\seqproject{\symEventLog}{1} \subseteq \symEventLog(\symContextPath')$
for the corresponding context path~$\symContextPath'$.

We only have to show that the introduction of the recursion operator 
maintains language rediscoverability (Lemma~12).
That is, assume:
1) a model $\symTree = \treeOpRecurOr_f$ adhering to the model restrictions; and 
2) $\symEventLog \subseteq \treeLan(\symTree) \land \symEventLog \opDFComplete \symTree$.
Then we have to show that we discover the model~$\symTree'$ such that $\symTree' = \symTree$.

Since we adhere to the model restrictions, due to restriction~4, we know there 
must be a larger model~$\symTree''$ such that the recursion node~$\treeOpRecurOr_f$ is a leaf of~$\symTree''$
and there exists a corresponding named subtree node~$\treeOpCompOr_f$ on the path from~$\symTree''$ to~$\treeOpRecurOr_f$.
Thus, we can conclude that $\symEventLog$ must be the sublog associated with 
a context path~$\symContextPath$ such that~$f \in \symContextPath$.
By code inspection on Algorithm~\ref{alg:recurse-discovery-run}, line~\ref{alg:recurse-discovery-run:check-recursion},
we see that we only have to prove that~$\forall x \in \symTrace \in \symEventLog : x[1] = f$.
This follows directly from~$\symEventLog \opDFComplete \symTree$.
Hence, the recursion operator is correctly rediscovered.

We can conclude that the Recursion Aware Discovery adaptation preserves language rediscoverability guarantee of~\cite[Theorem~14, Page~16]{sleemans2013-im}.
\end{proof}

% -- Content - Model Discovery - Guarantees Detail - Runtime Complexity --
\subsubsection{Runtime Complexity}

In~\cite[Run Time Complexity, Page~17]{sleemans2013-im}, the authors describe how
the basic discovery framework is implemented as a polynomial algorithm.
For the selection and log splitting for the normal process tree operators
(Alg.~\ref{alg:naive-discovery}, line~\ref{alg:naive-discovery:opsplit}, and
Alg.~\ref{alg:recurse-discovery-run}, line~\ref{alg:recurse-discovery-run:opsplit}),
existing polynomial algorithms were used.
Furthermore, for the original framework, the number of recursions made
is bounded by the number of activities: $O(\len{\Sigma(L)})$.
We will show that this polynomial runtime complexity is maintained for our adaptations.

In our Na\"ive Discovery adaptation, the number of recursions is determined
by Algorithm~\ref{alg:naive-discovery}, lines~\ref{alg:naive-discovery:subtree} and~\ref{alg:naive-discovery:opsplit}.
For line~\ref{alg:naive-discovery:subtree}, the number of recursions is bounded
by the depth of the hierarchical event log: $O(\|\symEventLog\|)$.
For line~\ref{alg:naive-discovery:opsplit}, the original number of activities bound holds: $O(\len{\Sigma(L)})$.
Thus, the total number of recursions for our Na\"ive Discovery is bounded by $O(\|\symEventLog\| + \len{\Sigma(L)})$.
Hence, the Na\"ive Discovery adaptation has a polynomial runtime complexity.

In one run of our Recursion Aware Discovery, the number of recursions is determined
by Algorithm~\ref{alg:recurse-discovery-run}, line~\ref{alg:recurse-discovery-run:opsplit}.
Note that the recursion and named subtree cases do not recurse directly due to the
delayed discovery principle.
For line~\ref{alg:recurse-discovery-run:opsplit}, the original number of activities bound holds: $O(\len{\Sigma(L)})$.
Thus, we can conclude that Algorithm~\ref{alg:recurse-discovery-run} has a polynomial runtime complexity.

For the complete Recursion Aware Discovery, the runtime complexity is determined
by Algorithm~\ref{alg:recurse-discovery}, lines~\ref{alg:recurse-discovery:iterate} and~\ref{alg:recurse-discovery:glue}.
Each iteration of the loop at line~\ref{alg:recurse-discovery:iterate} is polynomial.
The number of iterations is determined by the number of times an~$\symEventLog(\symContextPath)$ is changed.
Based on Algorithm~\ref{alg:recurse-discovery-run}, lines~\ref{alg:recurse-discovery-run:recursion} and~\ref{alg:recurse-discovery-run:subtree},
the number of times an~$\symEventLog(\symContextPath)$ is changed is bounded 
by the depth of the hierarchical event log: $O(\|\symEventLog\|)$.
Thus, the total number of iterations is polynomial and bounded by $O(\|\symEventLog\|)$.
Each iteration of the loop at line~\ref{alg:recurse-discovery:glue} is polinomial in the named tree depth, and thus bounded by $O(\|\symEventLog\|)$.
The number of iterations is determined by the number of named subtrees, and thus also bounded by $O(\|\symEventLog\|)$.
Hence, the Recursion Aware Discovery adaptation has a polynomial runtime complexity.

% -- Content - Postprocess --
\section{Using and Visualizing the Discovered Model}
\label{sec:postprocess}

Discovering a behavioral backbone model is only step one.
Equally important is how one is going to use the model, both for analysis and for further model driven engineering.
In this section, we touch upon some of the solutions we implemented, and demo in Subsection~\ref{sec:demo} and Figure~\ref{fig:demo}.

\subsection{Rewriting, Filtering and the 80/20 Model}
\label{sec:infreqmodel}

To help the user understand the logged behavior, we provide several ways of filtering the model, reducing the visible complexity, and adjusting the model quality.

Based on frequency information, we allow the user to inspect an 80/20 model.
An 80/20 model describes the mainstream (80\%) behavior using a simple (20\%) model~\cite{sleemans-thesis}.
We allow the user to interactively select the cutoff (80\% by default) using sliders directly next to the model visualization, thus enabling the ``real-time exploration'' of behavior.
Unusual behavior can be projected and highlighted onto the 80\% model using existing conformance and deviation detection techniques~\cite{adriansyah2014aligning}.
This way, it is immediately clear where the unusual behavior is present in the model, and how it is different from the mainstream behavior.

Based on hierarchical information, we allow both coarse and fine grained filtering.
Using sliders, the user can quickly select a minimum and maximum hierarchical depth to inspect, and hide other parts of the model.
The idea of depth filtering is illustrated in Figure~\ref{fig:filter:depth}.
Afterwards, users can interactively fold and unfold parts of the hierarchy.
By searching, users can quickly locate areas of interest.
Using term-based tree rewriting (see Table~\ref{tab:reduct-pt}), we present the user with a simplified model that preserves behavior.

\input{gfx/reduct/reduct-pt.tex}

\begin{figure}[!htb]%
    \centering%
    \vspace*{-0.5em}
    \noindent\adjustbox{max width=0.45\textwidth}{\input{gfx/reduct/reduct-depth.tex}}
    \vspace*{-0.5em}
    \caption{Illustration of depth filtering, where we hide everything above $x$ and below $y$.
    The dashed arrows relate the altered nodes.}%
    \label{fig:filter:depth}%
    \vspace*{-0.5em}
\end{figure}
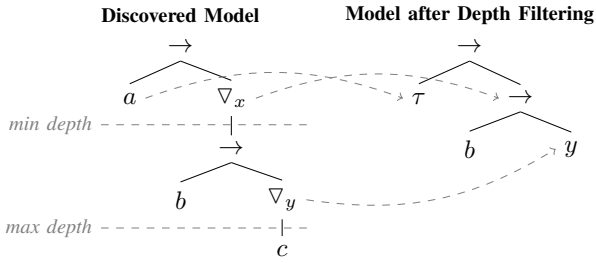% 

\begin{figure}[!htb]%
    \centering%
    \subfigure[Hierarchical process tree]{\adjustbox{max width=0.23\textwidth}{\input{gfx/map2/pt.tex}}}\hspace*{0.5cm}
    \subfigure[Statechart]{\adjustbox{max width=0.23\textwidth}{\input{gfx/map2/sc.tex}}}
    \subfigure[Sequence diagram]{\adjustbox{max width=0.23\textwidth}{\input{gfx/map2/sd.tex}}}\hspace*{0.5cm}
    \subfigure[Petri net]{\adjustbox{max width=0.23\textwidth}{\input{gfx/map2/ptnet.tex}}}
    \vspace*{-0.3em}
    \caption{A hierarchical process tree, and its mapping to different formalisms.}%
    \label{fig:map2}%
    \vspace*{-0.8em}
\end{figure}
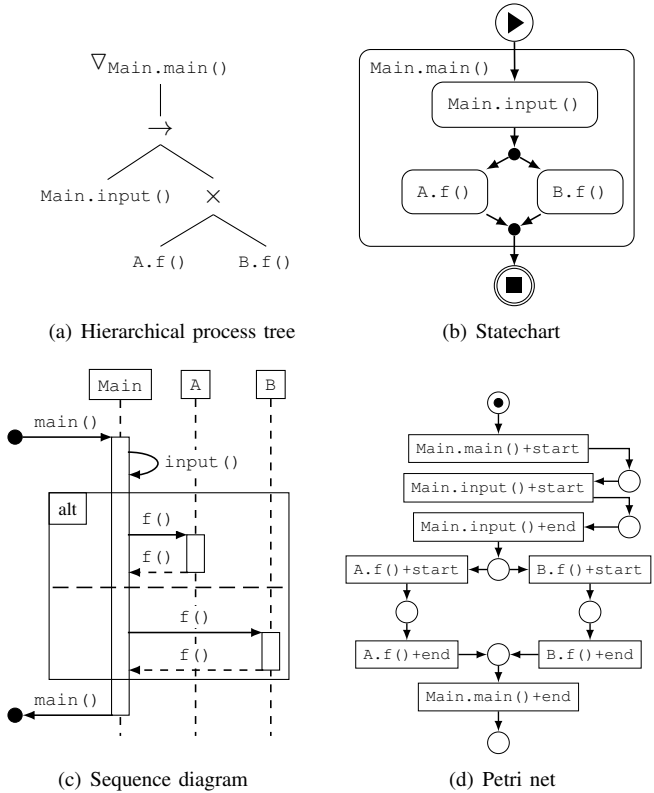% 

\subsection{Linking the Model to Event Data and the Source Code}
For further analysis, like performance (timing), frequency (usage) and conformance, we annotate the discovered model with additional information.
By annotating this information onto (parts of) the model, we can visualize it in the context of the behavior.
This annotation is based on the event log data provided as input and is partly provided by existing algorithms in the Process Mining Toolkit ProM.
Most notably, we 1) align the discovered model with the event log, as described in~\cite{adriansyah2014aligning}, and 2) link model elements back to the referenced source code lines in Eclipse that generated the logged events.

\subsection{Mapping and Visualizing}
For visualization and integration with existing techniques, we implemented mappings to other formalisms.
As noted in~\cite{sleemans-thesis}, process trees represent a ``block-structured'' language.
Hence, we can simply map each operator to a block-structured concept in the target formalism, and preserve behavior by construction.
We support mappings to formalisms such as Statecharts, (Data) Petri nets, Sequence diagrams, and BPMN diagrams.
Some example mappings are given in Figure~\ref{fig:map2}.

% -- Content - Evaluation2 --
\section{Evaluation}
\label{sec:evaluation}

In this section, we compare our technique against related, implemented techniques.
The proposed algorithms are implemented in the \emph{Statechart} plugin for the process mining framework ProM \cite{prom:statechart}.
In the remainder of this section, we will refer to Algorithm~\ref{alg:naive-discovery} as \emph{Na\"ive}, and to Algorithm~\ref{alg:recurse-discovery} as \emph{RAD} (short for Recursion Aware Discovery).
We end the evaluation by showing example results in our tool.

% -- Content - Evaluation2 - Methodology --
\subsection{Input and Methodology for Comparative Evaluation}
In the comparative evaluation, we focus on the quantitative aspects of the design criteria from the introduction.
That is, the approach and tools should provide a behavioral backbone model that is precise and fits the actual system.
We measure two aspects for a number of techniques and input event logs:
1) the \emph{running time} of the technique, and
2) the \emph{model quality}.

For the running time, we measured the average running time and associated 95\% confidence interval over 30 micro-benchmark executions, after 10 warmup rounds for the Java JVM.
Each technique is allowed at most 30 seconds for completing a single model discovery.
For the model quality, we use fitness and precision as described in~\cite{adriansyah2014aligning}, and set a time limit of at most 5 minutes.
In short, \emph{fitness} expresses the part of the log that is represented by the model; 
\emph{precision} expresses the behavior in the model that is present in the log.
For these experiments we used a laptop with an i7-4700MQ CPU @~2.40~GHz, Windows~8.1 and Java SE~1.7.0~67 (64~bit) with 12~GB of allocated RAM.

We selected five event logs as experiment input, covering a range of input problem sizes.
The input problem size is typically measured in terms of four metrics:
number of traces, number of events, number of activities (size of the alphabet), and average trace length.
The event logs and their sizes are shown in Table~\ref{tab:eval2:logs}.
The \emph{BPIC 2012}~\cite{xeslog:bpic2012} and \emph{BPIC 2013}~\cite{xeslog:bpic2013} event logs are so called BPI Challenge logs.
These large real-life event logs with complex behavior are often used in process mining evaluations.
The challenge logs are made available yearly in conjunction with the BPM conference and are considered sufficiently large and complex inputs to stress test process mining techniques.
The \emph{JUnit 4.12}~\cite{xeslog:junit412}, Apache Commons \emph{Crypto 1.0.0}~\cite{xeslog:apache-crypto}, and \emph{NASA CEV}~\cite{xeslog:nasa-cev} event logs are created using an extended version of the instrumentation tool developed for~\cite{leemans2015models}, yielding XES event logs with method-call level events.
The JUnit 4.12 software~\cite{sw:junit412} was executed once, using the example input found at~\cite{sw:junit-getting-started}.
For the Apache Commons Crypto 1.0.0 software~\cite{sw:apache-crypto}, we executed the \cf{CbcNoPaddingCipherStreamTest} unit test.
For the NASA CEV software~\cite{sw:nasa-cev}, we executed a unit test generated from the source code, covering all of the code branches.

\input{./gfx/eval2/log-info.tex}

\input{./gfx/eval3/result-time.tex}

We compare our discovery algorithms against a selection of the techniques mentioned in Section~\ref{sec:related-work}.
Unfortunately, we could not compare against some of the related work due to invalid input assumptions or the lack of a reference implementation.
The \emph{Inductive miner} (\emph{IM})~\cite{sleemans-thesis} is our baseline comparison algorithm, since our approach builds upon the IM framework.
For the Inductive miner and our derived techniques, we also consider the \emph{paths} setting.
This is the frequency cutoff for discovering an 80/20 model (see Subsection~\ref{sec:infreqmodel}): 1.0 means all behavior, 0.8 means 80\% of the behavior.
The work of~\cite{aalst2004workflow-alpha,Weijters2011-fhm,zelst2015-ilp-filter,vanderWerf2008-ilp} provides a comparison with standard process mining techniques.
The \emph{Fuzzy miner}~\cite{gunther2007fuzzy} provides a comparison for hierarchy discovery.
However, it yields models without semantics, and hence the quality cannot be determined.
The \emph{MINT algorithm (EFSM inference)}~\cite{Walkinshaw2016mint,sw:efsmi-mint} provides a comparison with the well-known \emph{redblue} and \emph{ktails} dynamic analysis techniques.
The \emph{Synoptic} algorithm~\cite{Beschastnikh2011synoptic,sw:synoptic} is an invariant based FSM inferrer.

For our techniques, we also consider the following \emph{heuristics for hierarchy}:
1) No heuristic (use log as is), 
2) Nested Calls, and 
3) Structured \cf{package.class.method} Names (Struct.
Names).
Note that the Nested Calls heuristic is only applicable for software event logs, and not for the BPIC logs.

% -- Content - Evaluation2 - Results and Discuss --
\subsection{Comparative Evaluation Results and Discussion}

\subsubsection{Runtime Analysis}
In Table~\ref{tab:eval2:results:time}, the results for the runtime benchmark are given.
As noted before, the Nested Calls setup is not applicable for the BPIC logs.
We immediately observe that the ILP, MINT, and Synoptic algorithms could not finish in time on most logs.
MINT and Synoptic have difficulty handling a large number of traces.
We also notice that most setups require a long processing time and a lot of memory for the Apache Crypto log.
Large trace lengths, such as in the Crypto log, are problematic for all approaches.
Our techniques overcome this problem by using the hierarchy to divide large traces into multiple smaller traces (see below).

When we take a closer look at the actual running times, we observe the advantages of the heuristics for hierarchy and accompanied discovery algorithms.
In all cases, using the right heuristic before discovery improves the running time.
In extreme cases, like the Apache Crypto log, it even makes the difference between getting a result and having no result at all.
Note that, with a poorly chosen heuristic, we might not get any improvements, e.g., note the absence of models for the Apache Crypto plus Structured Names heuristics.

The speedup factor for our technique depends on an implicit input problem size metric: the depth of the discovered hierarchy.
In Table~\ref{tab:eval2:results:depth}, the discovered depths are given for comparison.
For example, the dramatic decrease in running time for the JUnit log can be explained by the large depth in hierarchy: 25~levels in this case.
This implies that the event log is decomposed into many smaller sublogs, 
as per Alg.~\ref{alg:naive-discovery}, line~\ref{alg:naive-discovery:subtree}, 
and Alg.~\ref{alg:recurse-discovery-run}, lines~\ref{alg:recurse-discovery-run:recursion} and~\ref{alg:recurse-discovery-run:subtree}.
Hence, the imposed hierarchy indirectly yields a good decomposition of the problem, aiding the divide and conquer tactics of the underlying algorithms.

\input{./gfx/eval2/log-hierarchy.tex}

\input{./gfx/eval3/result-quality.tex}

\subsubsection{Model Quality Analysis}
In Table~\ref{tab:eval2:results:quality}, the results of the model quality measurements are given.
The Fuzzy miner is absent due to the lack of semantics for fuzzy models.
For the Structured Names heuristic, there was no trivial event mapping for the quality score measurement.

Compare the quality scores between the \emph{Nested Calls} setup with \emph{No heuristics} and the related work.
Note that in all cases, the Nested Calls setup yields a big improvement in precision, with no significant impact on fitness.
The drop in quality for the JUnit - RAD case is due to the limitations in translation of the recursive structure to the input format of~\cite{adriansyah2014aligning}.
In addition, our technique without heuristics, with paths at 1.0, maintains the model quality guarantees (perfect fitness).
Overall, we can conclude that the added expressiveness of modeling the hierarchy have a positive impact on the model quality.

% -- Content - Evaluation2 - Demo-small --

\begin{figure}[ht!]
    \centering
    \includegraphics[width=0.48\textwidth]{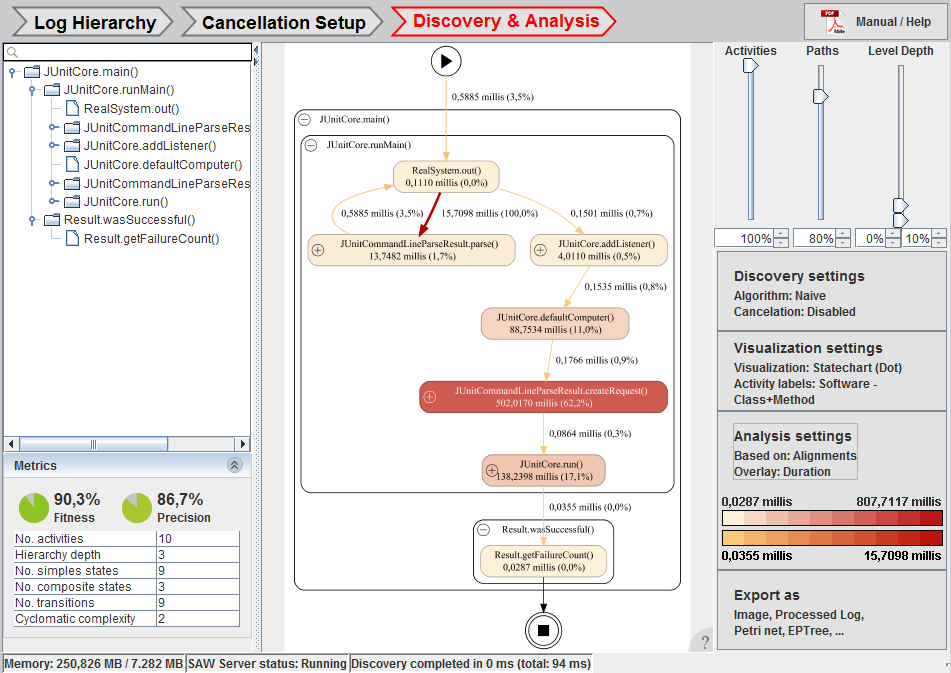}
    \vspace*{-1.5ex}
    \caption{Screenshot of the workbench tool implementation \cite{prom:statechart}}
    \label{fig:demo}
    \vspace*{-2.5ex}
\end{figure}

\subsection{Using the Tool}
\label{sec:demo}

In Figure~\ref{fig:demo}, we show some of the functionality and solutions available via the workbench tool implementation in~ProM~\cite{prom:statechart}.
Via the user interface, the analyst first selects the desired heuristic for hierarchy.
Afterwards, the analysis is presented through a Statechart visualization of the discovered model, annotated with frequency information.
With the sliders on the right and the tree view and search box on the left, the analyst can adjust parameters in real time and interactively explore the model.
The analyst can switch visualizations directly in the workbench UI, and view the model as, for example, a sequence diagram.
Thanks to a corresponding Eclipse plugin, the model and analysis results can be linked back to the source code.
A simple double-click on the model allows for a jump to the corresponding source code location.
In addition, we can overlay the code in Eclipse with our model analysis results.

% \begin{figure}[ht!]
%   \begin{center}
%     \subfigure[Log heuristics Selection screen]{
%       \label{fig:demo:logs}
%       \includegraphics[width=0.48\textwidth]
%         {gfx/demo/log-heuristic-3-s.png}}
%     \subfigure[Discovered model, visualized as a Statechart]{
%       \label{fig:demo:sc}
%       \includegraphics[width=0.48\textwidth]
%         {gfx/demo/discover-1-s.png}}
%     \subfigure[Discovered model, visualized as a Sequence Diagram]{
%       \label{fig:demo:sd}
%       \includegraphics[width=0.48\textwidth]
%         {gfx/demo/discover-3-2-s.png}}
%   \end{center}%
%   \vspace*{-2ex}
%   \caption{Screenshots of the workbench tool implementation \cite{prom:statechart}}
%   \label{fig:demo}
%   \vspace*{-3ex}
% \end{figure}

% -- Content - Conclusion-1 --
\section{Conclusion and Future Work}
\label{sec:conclusion}

In this paper, \emph{we 1) proposed a novel hierarchy and recursion extension to the process tree model; and 2) defined the first, recursion aware process model discovery technique that leverages hierarchical information in event logs, typically available for software system cases.}
This technique allows us to analyze the operational processes of software systems under real-life conditions at multiple levels of granularity.
An implementation  of the proposed algorithm is made available via the Statechart plugin in the process mining framework ProM~\cite{prom:statechart}.
The Statechart workbench provides an intuitive way to discover, explore and analyze hierarchical behavior, integrates with existing ProM plugins and links back to the source code in Eclipse.
Our experimental results, based on real-life (software) event logs, demonstrate the feasibility and usefulness of the approach and \emph{show the huge potential to speed up discovery by exploiting the available hierarchy}.

Future work aims to uncover exceptional and error control flows (i.e., try-catch and cancellation patterns), provide reliability analysis, and better support multi-threaded and distributed software.
In addition, enabling the proposed techniques in a streaming context could provide valuable real-time insight into software in its natural environment.
Furthermore, since (software) event logs can be very large, using a streaming context means we do not have to use a large amount of storage.

% -- References-report --
\bibliographystyle{plain}
\bibliography{ref}

% -- Footer --

\end{document}

%% file: gfx/relatedwork/tbl-compare.tex
\newcommand{\rwHDef}[1]{\multicolumn{1}{l}{#1}}
\newcommand{\rwHRot}[1]{\multicolumn{1}{c}{\hspace*{-1ex}\makebox[1em][r]{\rotatebox[origin=r]{-40}{\small{#1}}}}}

\newcommand{\rwCRule}{\cmidrule{2-17}}
\newcommand{\rwCRot}[2]{\multirow{#1}{*}{\rotatebox[origin=c]{90}{\scriptsize\textbf{#2}}}}

\begin{table*}[ht!]%
  \centering%
\noindent\adjustbox{max width=0.98\textwidth}{%
\begin{threeparttable}%
  \caption{Comparison of related techniques and the expressiveness of the
  resulting models, divided into the groups
  from Section~\ref{sec:related-work}.}%
  \vspace*{0.1cm}
    \begin{tabular}{
        @{\hskip3pt}l@{\hskip3pt}l@{\hskip3pt}l 
        | 
        l@{\hskip6pt}l@{\hskip6pt}l 
        |
        cc
        |
        ccc
        |
        cccccc
    }
          &       & \rwHDef{Author} 
          & \rwHDef{Technique / Toolkit} & \rwHDef{Input} 
          & \rwHDef{Formalism}
          & \rwHRot{Execution Semantics}
          & \rwHRot{Model Quality}
          & \rwHRot{Aggregate Runs}
          & \rwHRot{Frequency Info} 
          & \rwHRot{Performance Info}
          & \rwHRot{Choice} 
          & \rwHRot{Loop}
          & \rwHRot{Concurrency}
          & \rwHRot{Hierarchy}
          & \rwHRot{Named Submodels} 
          & \rwHRot{Recursion} \\
          
    \rwCRule
    \rwCRot{5}{Static Analysis}
          & \cite{tonella2003interact-c} & Tonella & Object flow analysis & C++ source code & UML Interact. & -     & n/a   & n/a   & -     & -     & -     & -     & -     & -     & -     & - \\
          & \cite{kollmann2001collabdiag} & Kollmann & Java code structures & Java source code & UML Collab. & -     & n/a   & n/a   & -     & -     & -     & -     & -     & -     & -     & - \\
          & \cite{korshunova2006cpp2xmi} & Korshunova & CPP2XMI, SQuADT & C++ source code & UML SD, AD & -     & n/a   & n/a   & -     & -     & \checkmark & \checkmark & -     & -     & -     & - \\
          & \cite{Rountev2005staticsd} & Rountev & Dataflow analysis & Java source code & UML SD & $\pm$\tnote{1} & n/a   & n/a   & -     & -     & \checkmark & \checkmark & -     & -     & -     & - \\
          & \cite{Amighi2012cfgjava} & Amighi & Sawja framework & Java byte code & CFG   & -     & n/a   & n/a   & -     & -     & \checkmark & \checkmark & -     & -     & -     & - \\

    \rwCRule
    \rwCRot{11}{Dynamic Analysis}
          & \cite{alalfi2009automated} & Alalfi & PHP2XMI & Instrumentation & UML SD & $\pm$\tnote{1} & -     & -     & -     & -     & -     & -     & -     & -     & -     & - \\
          & \cite{oechsle2002javavis} & Oechsle & JAVAVIS & Java debug interface & UML SD & $\pm$\tnote{1} & -     & -     & -     & -     & -     & -     & \checkmark & -     & -     & - \\
          & \cite{briand2003towards} & Briand & Meta models / OCL & Instrumentation & UML SD & $\pm$\tnote{1} & -     & -     & -     & -     & \checkmark & \checkmark & -     & -     & -     & - \\
          & \cite{briand2006toward} & Briand & Meta models / OCL & Instrumentation & UML SD & $\pm$\tnote{1} & -     & -     & -     & -     & \checkmark & \checkmark & \checkmark & -     & -     & - \\
          & \cite{labiche2013combining} & Labiche & Meta models / OCL & Instrumentation + source & UML SD & $\pm$\tnote{1} & -     & -     & -     & -     & \checkmark & \checkmark & -     & -     & -     & - \\
          & \cite{systa2001shimba} & Syst\"a & Shimba & Customized debugger & SD variant & $\pm$\tnote{1} & -     & -     & -     & -     & \checkmark & \checkmark & -     & -     & -     & - \\
          & \cite{Walkinshaw2016mint} & Walkinshaw & MINT  & Log with traces & EFSM  & \checkmark & \checkmark & \checkmark & -     & -     & \checkmark & \checkmark & -     & -     & -     & - \\
          & \cite{beschastnikh2014inferring} & Beschastnikh & CSight & Given log, instrument & CFSM  & \checkmark & \checkmark & \checkmark & -     & -     & \checkmark & \checkmark & \checkmark & -     & -     & - \\
          & \cite{ackermann2009recovering} & Ackermann & Behavior extraction & Monitor network packets & UML SD & $\pm$\tnote{1} & -     & \checkmark & -     & \checkmark & \checkmark & \checkmark & -     & -     & -     & - \\
          & \cite{Graham1982gprof} & Graham & gprof profiler & Instrumentation & Call graphs & -     & -     & -     & -     & \checkmark & -     & -     & -     & \checkmark & \checkmark & $\pm$\tnote{8} \\
          & \cite{DePauw1998execpattern} & De Pauw & Execution patterns & Program Trace & Exec. Pattern & -     & -     & -     & \checkmark & -     & -     & \checkmark & -     & \checkmark & \checkmark & \checkmark \\
          & \cite{Beschastnikh2011synoptic} & Beschastnikh & Synoptic & Log with traces & FSM   & \checkmark & \checkmark & \checkmark & \checkmark & -     & \checkmark & \checkmark & -     & -     & -     & - \\
          & \cite{Heule2010} & Heule & DFAsat & Log with traces & DFA   & \checkmark & \checkmark & \checkmark & -     & -     & \checkmark & \checkmark & -     & -     & -     & - \\
    
    \rwCRule
    \rwCRot{3}{Grammar}
          & \cite{NevillManning1997sequitur} & Nevill-Manning & Sequitur & Symbol sequence & Grammar & \checkmark & -     & -     & -     & -     & -     & -     & -     & \checkmark & -     & - \\
          & \cite{Siyari2016lexis} & Siyari & Lexis & Symbol sequence & Lexis-DAG & \checkmark & -     & \checkmark & -     & -     & -     & -     & -     & \checkmark & -     & - \\
          & \cite{Jonyer2004subduegl} & Jonyer & SubdueGL & Symbol Graph & Graph Grammar & \checkmark & -     & -     & -     & -     & \checkmark & -     & -     & \checkmark & -     & $\pm$\tnote{9} \\

    \rwCRule
    \rwCRot{8}{Process Mining}
          & \cite{aalst2004workflow-alpha} & Van der Aalst & Alpha algorithm & Event Log & Petri net & \checkmark & -     & \checkmark & \checkmark\tnote{2} & \checkmark\tnote{2} & \checkmark & \checkmark & \checkmark & -     & -     & - \\
          & \cite{aalst2010process} & Van der Aalst & Theory of Regions & Event Log & Petri net & \checkmark & -     & \checkmark & \checkmark\tnote{2} & \checkmark\tnote{2} & \checkmark & \checkmark & \checkmark & -     & -     & - \\
          & \cite{Weijters2011-fhm} & Weijters  & Flexible heuristics miner & Event Log & Heuristics net & \checkmark & -     & \checkmark & \checkmark\tnote{2} & \checkmark\tnote{2} & \checkmark & \checkmark & \checkmark & -     & -     & - \\
          & \cite{vanderWerf2008-ilp} & Werf, van der & ILP miner & Event Log & Petri net & \checkmark & \checkmark & \checkmark & \checkmark\tnote{2} & \checkmark\tnote{2} & \checkmark & \checkmark & \checkmark & -     & -     & - \\
          & \cite{zelst2015-ilp-filter} & Zelst, S. J. van & ILP with filtering & Event Log & Petri net & \checkmark & \checkmark & \checkmark & \checkmark\tnote{2} & \checkmark\tnote{2} & \checkmark & \checkmark & \checkmark & -     & -     & - \\
          & \cite{aalst2005genetic} & Alves de Medeiros & Genetic Miner & Event Log & Petri net & \checkmark & \checkmark & \checkmark & \checkmark\tnote{2} & \checkmark\tnote{2} & \checkmark & \checkmark & \checkmark & -     & -     & - \\
          & \cite{buijs2012role} & Buijs & ETM algorithm & Event Log & Process tree & \checkmark & \checkmark & \checkmark & \checkmark\tnote{2} & \checkmark\tnote{2} & \checkmark & \checkmark & \checkmark & -     & -     & - \\
          & \cite{sleemans-thesis} & Leemans S.J.J. & Inductive Miner & Event Log & Process tree & \checkmark & \checkmark & \checkmark & \checkmark\tnote{2} & \checkmark\tnote{2} & \checkmark & \checkmark & \checkmark & -     & -     & - \\
          & \cite{gunther2007fuzzy} & G\"unther & Fuzzy Miner & Event Log & Fuzzy model & -     & -     & \checkmark & \checkmark\tnote{3} & -     & \checkmark & \checkmark\tnote{4} & -     & \checkmark & -     & - \\
          & \cite{Bose2012HPM} & Bose  & Two-phase discovery & Event Log & Fuzzy model & -     & -     & \checkmark & \checkmark\tnote{3} & -     & \checkmark & \checkmark\tnote{5} & -     & \checkmark & -     & - \\
          & \cite{conforti2016bpmnminer} & Conforti & BPMN miner & Event Log & BPMN  & \checkmark & \checkmark  & \checkmark & $\pm$\tnote{2} & $\pm$\tnote{2} & \checkmark & \checkmark\tnote{6} & \checkmark & \checkmark & -     & - \\
    
    \rwCRule
          &       & \textbf{This paper} & Recursion Aware Disc. & Event Log & H. Process tree\tnote{*} & \checkmark & \checkmark & \checkmark & \checkmark\tnote{2} & \checkmark\tnote{2} & \checkmark & \checkmark\tnote{7} & \checkmark & \checkmark & \checkmark & \checkmark \\
    \end{tabular}%
    \vspace*{-1em}
    \begin{tablenotes}%
    \begin{multicols}{2}
        \item[1] Formal semantics are available for UML SD variants.
        \item[2] Aligning an event log and a process model enables advanced performance, frequency, and conformance analysis, as described in~\cite{adriansyah2014aligning,sleemans-thesis}.
        \item[3] Various log-based process metrics have been defined which capture different notions of frequency, significance, and correlation~\cite{gunther2007fuzzy}.
        \item[4] The hierarchy is based on anonymous clusters in the resulting model~\cite{gunther2007fuzzy}.
        \item[5] The hierarchy is based on abstraction patterns over events~\cite{Bose2009AbstractTax, Bose2012HPM}.
        \item[6] The hierarchy is based on discovered relations over extra data in the event log.
        \item[7] The hierarchy is based on the hierarchical information in the event log.
        \item[8] Recursion is detectable as a cycle, but without performance analysis support.
        \item[9] Only tail recursion is supported.
        \item[*] \emph{Hierarchical process tree}, as introduced in Definition~\ref{def:hier-process-tree}.
     \end{multicols}
     \end{tablenotes}%
  \label{tab:related-work-compare}%
  \end{threeparttable}%
}
  \vspace*{-2.2em}
\end{table*}%

%% file: gfx/model/sw-code.tex
\vspace*{0.2cm}%
\begin{lstlisting}[language=Java, basicstyle=\scriptsize\ttfamily,
frame=single, numbers=left, xleftmargin=6.0ex, xrightmargin=3.0ex, captionpos=b,
label=lst:running-example, 
caption={Running example program code, logged at the method level.}]
public class Main {
    public static void main(int argument) {
        A inst = input(argument);
        inst.process(argument);
        output();
    }
    private static A input(int i) { ... }
    private static void output() { ... }
}
class A {
    public void process(int i) { ... }
}
class B extends A {
    public void process(int i) {
        if (i <= 0) {
            super.process(i);
        } else {
            stepPre();
            process(i - 1);
            stepPost();
        }
    }
    private void stepPre() { ... }
    private void stepPost() { ... }
}
\end{lstlisting}
\vspace*{-0.1cm}

%% file: gfx/model/sw-ptree.tex
\begin{tikzpicture}[
    sibling distance = 6em,
    level distance = 2em,
    every node/.style = {draw=none, fill=none, align=center}
]
  \node {$\treeOpCompOr_{\text{\cf{Main.main()}}}$}
    child { node {$\treeOpSeq$}
      child { node {\cf{Main.input()}} }
      child { node {$\treeOpCompOr_{\text{\cf{B.process()}}}$}
        child { node {$\treeOpXor$}
           child { node {\cf{A.process()}} }
           child { node {$\treeOpSeq$}
             child { node {\cf{B.stepPre()}} }
             child { node {$\treeOpRecurOr_{\text{\cf{B.process()}}}$} }
             child { node {\cf{B.stepPost()}} }
           }
        }
      }
      child { node {\cf{Main.output()}} } 
    };
\end{tikzpicture}

%% file: gfx/model/sw-log-int.tex
\begin{tikzpicture}[
    node distance=0.5cm,
    event/.style={diamond,fill, scale=0.4, text width=0.1cm},
    edge/.style={semithick}
    label/.style={}
]
%legend
    \node[label] at (7.3,0.1) {\footnotesize \textbf{Legend}};
    
   \node[event] (start_s) at (7, -0.3) {};
   \node[] (start_c) at (7.4, -0.3) {};
   \draw[edge] (start_s) to (start_c);
   \node[label] at (8.0,-0.3) {\footnotesize start event};
   
   \node[] (end_s) at (6.8, -0.7) {};
   \node[event] (end_c) at (7.2, -0.7) {};
   \draw[edge] (end_s) to (end_c);
   \node[label] at (8.0,-0.7) {\footnotesize end event};

% example
    % Nodes
    % l0
    \node[event] (main_s) at (0,0) {};
    \node[event] (main_c) at (6,0) {};
    \node[label] at (0.7,0.25) {\footnotesize Main.main()};
    
    %l1
    \node[event] (input_s) at (0.1,-0.65) {};
    \node[event] (input_c) at (1.2,-0.65) {};
    \node[label] at (0.7,-0.4) {\footnotesize Main.input()};
    
    \node[event] (Bproc1_s) at (1.6,-0.65) {};
    \node[event] (Bproc1_c) at (4.4,-0.65) {};
    \node[label] at (2.2,-0.4) {\footnotesize B.process()};
    
    \node[event] (output_s) at (4.8,-0.65) {};
    \node[event] (output_c) at (5.9,-0.65) {};
    \node[label] at (5.5,-0.4) {\footnotesize Main.output()};
    
    %l2
    \node[event] (spre_s) at (1.7,-1.25) {};
    \node[event] (spre_c) at (2.0,-1.25) {};
    \node[label] at (1.5,-1) {\footnotesize B.stepPre()};
    
    \node[event] (Bproc2_s) at (2.4,-1.25) {};
    \node[event] (Bproc2_c) at (3.6,-1.25) {};
    \node[label] at (3.0,-1) {\footnotesize B.process()};
    
    \node[event] (spost_s) at (4.0,-1.25) {};
    \node[event] (spost_c) at (4.3,-1.25) {};
    \node[label] at (4.6,-1) {\footnotesize B.stepPost()};
    
    %l3
    \node[event] (Aproc_s) at (2.6,-1.85) {};
    \node[event] (Aproc_c) at (3.4,-1.85) {};
    \node[label] at (3.1,-1.6) {\footnotesize A.process()};
    
    % Edges
    \draw[edge] (main_s) to (main_c);
    
    \draw[edge] (input_s) to (input_c);
    \draw[edge] (Bproc1_s) to (Bproc1_c);
    \draw[edge] (output_s) to (output_c);
    
    \draw[edge] (spre_s) to (spre_c);
    \draw[edge] (Bproc2_s) to (Bproc2_c);
    \draw[edge] (spost_s) to (spost_c);
    
    \draw[edge] (Aproc_s) to (Aproc_c);

% % Nodes
%   \node[event] (f0Start) at (0,0) {};
%   \node [above of=f0Start, text width=3cm] {\footnotesize
%     $\symClassifier(e_0) = \text{\cf{B.process()}}$
%     $\symClassifier_\ell(e_0) = \text{start}$
%   };
%   \node[event] (f0End) at (5,0) {};
%   \node [above of=f0End, text width=3cm] {\footnotesize
%     $\symClassifier(e_1) = \text{\cf{B.process()}}$
%     $\symClassifier_\ell(e_1) = \text{complete}$
%   };
%   
%   
%   \node[event] (f1Start) at (0.5,-0.4) {};
%   \node [below of=f1Start, text width=3cm] {\footnotesize
%     $\symClassifier(e_2) = \text{\cf{A.process()}}$
%     $\symClassifier_\ell(e_2) = \text{start}$
%   };
%   \node[event] (f1End) at (4.5,-0.4) {};
%   \node [below of=f1End, text width=3cm] {\footnotesize
%     $\symClassifier(e_3) = \text{\cf{A.process()}}$
%     $\symClassifier_\ell(e_3) = \text{complete}$
%   };
%   
% % Edges
%   \draw[edge] (f0Start) to (f0End);
%   \draw[edge] (f1Start) to (f1End);
\end{tikzpicture}

%% file: gfx/model/sw-log-2.tex
\newcommand{\cfs}[1]{#1}

\begin{table}[ht!]
  \centering
  \caption{A single trace in an example event log corresponding to the program
  in Listing~\ref{lst:running-example} and the intervals in
  Figure~\ref{fig:ex-log-int}.
  Each column is one event, and each row a level in the hierarchy.
  }
  \vspace*{-1.5ex}
    \begin{tabular}{
    |@{\hskip4pt}c@{\hskip4pt}
    |@{\hskip4pt}c@{\hskip4pt}
    |@{\hskip4pt}c@{\hskip4pt}
    |@{\hskip4pt}c@{\hskip4pt}
    |@{\hskip4pt}c@{\hskip4pt}|}
    \hline
    \cfs{Main.main()}    & \cfs{Main.main()}  & \cfs{Main.main()} & \cfs{Main.main()}  & \cfs{Main.main()}
    \Tstrut\Bstrut\\ \hline
    \cfs{Main.input()}   & \cfs{B.process()}  & \cfs{B.process()} & \cfs{B.process()}  & \cfs{Main.output()}
    \Tstrut\Bstrut\\ \hline
                         & \cfs{B.stepPre()}  & \cfs{B.process()} & \cfs{B.stepPost()} &
    \Tstrut\Bstrut\\ \hline
                         &                    & \cfs{A.process()} &                    &
    \Tstrut\Bstrut\\ \hline
    \end{tabular}%
  \label{tab:ex-htrace}%
  \vspace*{-3ex}
\end{table}%

%% file: gfx/disc/exb-im-table.tex
\begin{table}[htbp]
  \centering
  \caption{Example Discovery on the log $\mset{\seq{a,b,d}, \seq{a,c,d,e,d}}$.
  The rows illustrate how the discovery progresses.
  The highlights indicate the sublogs used, and relate them to the corresponding
  partial model that is discovered. }
   \vspace*{-0.3cm}
   \begin{tabular}{M{3mm} M{40mm} M{32mm}}
    \toprule
        \textbf{Step} 
      & \textbf{Discovered Model}
      & \textbf{Event Log}
      \\
    \midrule
         1 & 
         \raisebox{-.8\totalheight}{\input{./gfx/disc/exb-im-s0-pt.tex}}
         &
         \input{./gfx/disc/exb-im-s0-log.tex}
      \\
    \midrule
         2 & 
         \raisebox{-.8\totalheight}{\input{./gfx/disc/exb-im-s1-pt.tex}}
         &
         \input{./gfx/disc/exb-im-s1-log.tex}
      \\
    \midrule
         3 & 
         \raisebox{-.8\totalheight}{\input{./gfx/disc/exb-im-s2-pt.tex}}
         &
         \input{./gfx/disc/exb-im-s2-log.tex}
      \\
    \midrule
         4 & 
         \raisebox{-.8\totalheight}{\input{./gfx/disc/exb-im-s3-pt.tex}}
         &
         \input{./gfx/disc/exb-im-s3-log.tex}
      \\
    \bottomrule
   \end{tabular}
   \vspace*{-0.3cm}
  \label{tab:disc-exb-im}%
\end{table}%

%% file: gfx/disc/exb-im-s0-pt.tex
\begin{tikzpicture}[
    level 1/.style={sibling distance=5em},
    level 2/.style={sibling distance=3em},
    level 3/.style={sibling distance=1em},
    level distance = 2.2em,
    every node/.style = {draw=none, fill=none, align=center},
    sel/.style={above right,fill=ELsel}
]
  \node (seq) {$\treeOpSeq$}
      child { node {$\treeOpPlaceholder$} }
      child { node {$\treeOpPlaceholder$} }
      child { node {$\treeOpPlaceholder$} }
    ;
    
  \begin{pgfonlayer}{background}
    %\node[sel,fit=(seq)(a)] (highlight) {};
    \node[sel,fit=(seq)] (highlight1) {};
    %\node[sel,fit=(a)] (highlight2) {};
  \end{pgfonlayer}
\end{tikzpicture}

%% file: gfx/disc/exb-im-s0-log.tex
\begin{tabular}{|c|c|c|c|c|}
    \cline{1-3}
    \cellcolor{ELsel}$a$ & \cellcolor{ELsel}$b$ & \cellcolor{ELsel}$d$ \\
    \cline{1-3}
    \multicolumn{1}{c}{$\,$} \\
    \hline
    \cellcolor{ELsel}$a$ & \cellcolor{ELsel}$c$ & \cellcolor{ELsel}$d$
    & \cellcolor{ELsel}$e$ & \cellcolor{ELsel}$d$  \\
    \hline
\end{tabular}%

%% file: gfx/disc/exb-im-s1-pt.tex
\begin{tikzpicture}[
    level 1/.style={sibling distance=5em},
    level 2/.style={sibling distance=3em},
    level 3/.style={sibling distance=1em},
    level distance = 2.2em,
    every node/.style = {draw=none, fill=none, align=center},
    sel/.style={above right,fill=ELsel}
]
  \node (seq) {$\treeOpSeq$}
      child { node (a) {$a$} }
      child { node {$\treeOpPlaceholder$} }
      child { node {$\treeOpPlaceholder$} }
    ;
    
  \begin{pgfonlayer}{background}
    %\node[sel,fit=(seq)(a)] (highlight) {};
    %\node[sel,fit=(seq)] (highlight1) {};
    \node[sel,fit=(a)] (highlight2) {};
  \end{pgfonlayer}
\end{tikzpicture}

%% file: gfx/disc/exb-im-s1-log.tex
\begin{tabular}{|c|c|c|c|c|}
    \cline{1-3}
    \cellcolor{ELsel}$a$ & $b$ & $d$ \\
    \cline{1-3}
    \multicolumn{1}{c}{$\,$} \\
    \hline
    \cellcolor{ELsel}$a$ & $c$ & $d$ & $e$ & $d$  \\
    \hline
\end{tabular}%

%% file: gfx/disc/exb-im-s2-pt.tex
\begin{tikzpicture}[
    level 1/.style={sibling distance=5em},
    level 2/.style={sibling distance=3em},
    level 3/.style={sibling distance=1em},
    level distance = 2.2em,
    every node/.style = {draw=none, fill=none, align=center},
    sel/.style={above right,fill=ELsel}
]
  \node (seq) {$\treeOpSeq$}
      child { node (a) {$a$} }
      child { node (xor) {$\treeOpXor$}
        child { node (b) {$b$} }
        child { node (c) {$c$} }
      }
      child { node {$\treeOpPlaceholder$} }
    ;
    
  \begin{pgfonlayer}{background}
    \node[sel,fit=(xor)(b)(c)] (highlight) {};
  \end{pgfonlayer}
\end{tikzpicture}

%% file: gfx/disc/exb-im-s2-log.tex
\begin{tabular}{|c|c|c|c|c|}
    \cline{1-3}
    $a$ & \cellcolor{ELsel}$b$ & $d$ \\
    \cline{1-3}
    \multicolumn{1}{c}{$\,$} \\
    \hline
    $a$ & \cellcolor{ELsel}$c$ & $d$ & $e$ & $d$  \\
    \hline
\end{tabular}%

%% file: gfx/disc/exb-im-s3-pt.tex
\begin{tikzpicture}[
    level 1/.style={sibling distance=5em},
    level 2/.style={sibling distance=3em},
    level 3/.style={sibling distance=1em},
    level distance = 2.2em,
    every node/.style = {draw=none, fill=none, align=center},
    sel/.style={above right,fill=ELsel}
]
  \node (seq) {$\treeOpSeq$}
      child { node (a) {$a$} }
      child { node {$\treeOpXor$}
        child { node (b) {$b$} }
        child { node (c) {$c$} }
      }
      child { node (loop) {$\treeOpLoop$} 
        child { node (d) {$d$} }
        child { node (e) {$e$} }
      }
    ;
    
  \begin{pgfonlayer}{background}
    \node[sel,fit=(loop)(d)(e)] (highlight) {};
  \end{pgfonlayer}
\end{tikzpicture}

%% file: gfx/disc/exb-im-s3-log.tex
\begin{tabular}{|c|c|c|c|c|}
    \cline{1-3}
    $a$ & $b$ & \cellcolor{ELsel}$d$ \\
    \cline{1-3}
    \multicolumn{1}{c}{$\,$} \\
    \hline
    $a$ & $c$ & \cellcolor{ELsel}$d$ & \cellcolor{ELsel}$e$ &
    \cellcolor{ELsel}$d$  \\
    \hline
\end{tabular}%

%% file: gfx/disc/alg-naive.tex
\vspace*{-0.8ex}%
\begin{algorithm}
    \algfontsize
    \caption{Naive Discovery (\emph{Naive})\label{alg:naive-discovery}}
    \alginout
        {A hierarchical event log $\symEventLog$}
        {A hierarchical process tree $\symTree$ such that $\symEventLog$ fits
        $\symTree$} \algdescript{Extended framework, using the named subtree operator~$\treeOpCompOr_f$.}
    \algname{$Naive$}{$\symEventLog$}
\begin{algtab}
  \algif{$\forall \symTrace \in \symEventLog : \symTrace = \symEmptyTrace$}
    \algreturn $\actSilent$ \hspace*{1em}\emph{// the log is empty or only contains empty traces} \\
  \algelseif{$\exists f \in \actAlph : \forall \symTrace \in \symEventLog : \symTrace = \seq{f}$}
    \algreturn $f$ \hspace*{1em}\emph{// the log only has a single low-level activity} \\
  \algelseif{$\exists f \in \actAlph : (\forall x \in t \in \symEventLog : x[1] = f) \land (\exists x \in t \in \symEventLog : \len{x} > 1)$}
    \emph{// all events start with $f$, and there is a lower level in the hierarchy} \\
    \algreturn $\treeOpCompOr_f(\funn{Naive}(\seqproject{\symEventLog}{1}))$\label{alg:naive-discovery:subtree}\\
  \algelse
    \emph{// normal framework cases, based on the tree semantics (Def.~\ref{def:processtree})} \\
    \label{alg:naive-discovery:opsplit}Split $L$ into into sublogs $\symEventLog_1, \ldots, \symEventLog_n$, such
    that: \\
    $\exists \treeOp \in \set{\treeOpSeq, \treeOpXor, \treeOpLoop, \treeOpPar}
    : L \subseteq \treeOp^l(L_1, \ldots, L_n)$ \hspace*{1em}\emph{// see~\cite{sleemans-thesis}}\\
    \algreturn $\treeOp(\funn{Naive}(L_1), \ldots, \funn{Naive}(L_n))$ \\
  \algend
\end{algtab}
\end{algorithm}
\vspace*{-0.8ex}%

%% file: gfx/disc/ex3-nd-table.tex
\begin{table}[htbp]
  \centering
  \caption{Example Naive Discovery on the log $\mset{\seq{f.a, f.f.b}}$.
  The rows illustrate how the discovery progresses.
  The highlights indicate the sublogs used, and relate them to the corresponding
  partial model that is discovered. 
  }
   \vspace*{-0.3cm}
   \begin{tabular}{M{3mm} M{22mm} M{20mm} M{24mm}}
    \toprule
        \textbf{Step} 
      & \textbf{Discovered Model}
      & \textbf{Event Log}
      & \textbf{Sublog View}
      \\
    \midrule
         1 & 
         \raisebox{-.8\totalheight}{\input{./gfx/disc/ex3-nd-s1-pt.tex}}
         &
         \input{./gfx/disc/ex3-nd-s1-log.tex}
         &
         $\mset{ \seq{ f.a , f.f.b } }$
      \\
    \midrule
         2 &
         \raisebox{-.8\totalheight}{\input{./gfx/disc/ex3-nd-s2-pt.tex}}
         &
         \input{./gfx/disc/ex3-nd-s2-log.tex}
         &
         $\mset{ \seq{ a , f.b } }$
      \\
    \midrule
         3 &
         \raisebox{-.8\totalheight}{\input{./gfx/disc/ex3-nd-s3-pt.tex}}
         &
         \input{./gfx/disc/ex3-nd-s3-log.tex}
         &
         $\mset{ \seq{ b } }$
      \\
    \bottomrule
   \end{tabular}
   \vspace*{-0.1cm}
  \label{tab:disc-ex3-nd}%
\end{table}%

%% file: gfx/disc/ex3-nd-s1-pt.tex
\begin{tikzpicture}[
    sibling distance = 3em,
    level distance = 2.1em,
    every node/.style = {draw=none, fill=none, align=center},
    sel/.style={above right,fill=ELsel}
]
  \node (f1) {$\treeOpCompOr_{f}$}
    child { node {$\treeOpPlaceholder$} 
    };
    
  \begin{pgfonlayer}{background}
    \node[sel,fit=(f1)] (highlight) {};
  \end{pgfonlayer}
\end{tikzpicture}

%% file: gfx/disc/ex3-nd-s1-log.tex
\begin{tabular}{|c|c|}
    \hline
    \cellcolor{ELsel}$f$     & \cellcolor{ELsel}$f$ \\
    \hline
    $a$     & $f$ \\
    \hline
            & $b$ \\
    \hline
\end{tabular}%

%% file: gfx/disc/ex3-nd-s2-pt.tex
\begin{tikzpicture}[
    sibling distance = 3em,
    level distance = 2.1em,
    every node/.style = {draw=none, fill=none, align=center},
    sel/.style={above right,fill=ELsel}
]
  \node (f1) {$\treeOpCompOr_{f}$}
    child { node (seq) {$\treeOpSeq$} 
      child { node (a) {$a$} }
      child { node (f2) {$\treeOpCompOr_{f}$}
        child { node {$\treeOpPlaceholder$} }
      }
    };
    
  \begin{pgfonlayer}{background}
    \node[sel,fit=(seq)(a)(f2)] (highlight) {};
  \end{pgfonlayer}
\end{tikzpicture}

%% file: gfx/disc/ex3-nd-s2-log.tex
\begin{tabular}{|c|c|}
    \hline
    $f$     & $f$ \\
    \hline
    \cellcolor{ELsel}$a$     & \cellcolor{ELsel}$f$ \\
    \hline
            & $b$ \\
    \hline
\end{tabular}%

%% file: gfx/disc/ex3-nd-s3-pt.tex
\begin{tikzpicture}[
    sibling distance = 3em,
    level distance = 2.1em,
    every node/.style = {draw=none, fill=none, align=center},
    sel/.style={above right,fill=ELsel}
]
  \node (f1) {$\treeOpCompOr_{f}$}
    child { node (seq) {$\treeOpSeq$} 
      child { node (a) {$a$} }
      child { node (f2) {$\treeOpCompOr_{f}$}
        child { node (b) {$b$} }
      }
    };
    
  \begin{pgfonlayer}{background}
    \node[sel,fit=(b)] (highlight) {};
  \end{pgfonlayer}
\end{tikzpicture}

%% file: gfx/disc/ex3-nd-s3-log.tex
\begin{tabular}{|c|c|}
    \hline
    $f$     & $f$ \\
    \hline
    $a$     & $f$ \\
    \hline
            & \cellcolor{ELsel}$b$ \\
    \hline
\end{tabular}%

%% file: gfx/disc/ex3-rd-table.tex
\begin{table}[htbp]
  \centering
  \caption{Example Recursion Aware Discovery on the log $\mset{\seq{f.a,
  f.f.b}}$ The rows illustrate how the discovery progresses.
  The highlights indicate the sublogs used, and relate them to the corresponding
  partial model that is discovered. 
  }
   \vspace*{-0.3cm}
   \begin{tabular}{M{3mm} M{22mm} M{20mm} M{24mm}}
    \toprule
        \textbf{Step} 
      & \textbf{Discovered Model}
      & \textbf{Event Log}
      & \textbf{Sublog View}
      \\
    \midrule
         1 & 
         \input{./gfx/disc/ex3-rd-s1-pt.tex}
         &
         \input{./gfx/disc/ex3-rd-s1-log.tex}
         &
         \vspace*{-0.6cm}
         $$L = $$\vspace*{-0.5cm}
         $$\mset{ \seq{ f.a , f.f.b } }$$
         (Context~$\symContextPath = \symEmptyTrace$)
      \\
    \midrule
         2 &
         \input{./gfx/disc/ex3-rd-s2-pt.tex}
         &
         \input{./gfx/disc/ex3-rd-s2-log.tex}
         &
         \vspace*{-0.6cm}
         $$L(\seq{f}) = $$\vspace*{-0.5cm}
         $$\mset{ \seq{ a , f.b } }$$
         (Context~$\symContextPath = \seq{f}$)
      \\
    \midrule
         3 &
         \input{./gfx/disc/ex3-rd-s3-pt.tex}
         &
         \input{./gfx/disc/ex3-rd-s3-log.tex}
         &
         \vspace*{-0.6cm}
         $$L(\seq{f}) = $$\vspace*{-0.5cm}
         $$\mset{ \seq{ a , f.b }, \seq{ b } }$$
         (Context~$\symContextPath = \seq{f}$)
      \\
    \bottomrule
   \end{tabular}
   \vspace*{-0.3cm}
  \label{tab:disc-ex3-rd}%
\end{table}%

%% file: gfx/disc/ex3-rd-s1-pt.tex
\begin{tikzpicture}[
    sibling distance = 3em,
    level distance = 2.1em,
    every node/.style = {draw=none, fill=none, align=center},
    sel/.style={above right,fill=ELsel}
]
  \node (f1) {$\treeOpCompOr_{f}$}
    child { node {$\treeOpPlaceholder$} 
    };
    
  \begin{pgfonlayer}{background}
    \node[sel,fit=(f1)] (highlight) {};
  \end{pgfonlayer}
\end{tikzpicture}

%% file: gfx/disc/ex3-rd-s1-log.tex
\begin{tabular}{|c|c|}
    \hline
    \cellcolor{ELsel}$f$     & \cellcolor{ELsel}$f$ \\
    \hline
    $a$     & $f$ \\
    \hline
            & $b$ \\
    \hline
\end{tabular}%

%% file: gfx/disc/ex3-rd-s2-pt.tex
\begin{tikzpicture}[
    sibling distance = 3em,
    level distance = 2.1em,
    every node/.style = {draw=none, fill=none, align=center},
    sel/.style={above right,fill=ELsel}
]
  \node (f1) {$\treeOpCompOr_{f}$}
    child { node (seq) {$\treeOpSeq$} 
      child { node (a) {$a$} }
      child { node (f2) {$\treeOpRecurOr_{f}$}
      }
    };
    
  \begin{pgfonlayer}{background}
    \node[sel,fit=(seq)(a)(f2)] (highlight) {};
  \end{pgfonlayer}
\end{tikzpicture}

%% file: gfx/disc/ex3-rd-s2-log.tex
\begin{tabular}{|c|c|}
    \hline
    $x$     & $f$ \\
    \hline
    \cellcolor{ELsel}$a$     & \cellcolor{ELsel}$f$ \\
    \hline
            & $b$ \\
    \hline
\end{tabular}%

%% file: gfx/disc/ex3-rd-s3-pt.tex
\begin{tikzpicture}[
    sibling distance = 3em,
    level distance = 2.1em,
    every node/.style = {draw=none, fill=none, align=center},
    sel/.style={above right,fill=ELsel}
]
  \node (f1) {$\treeOpCompOr_{f}$}
    child { node (xor) {$\treeOpXor$}
      child { node (b) {$b$} }
      child { node (seq) {$\treeOpSeq$} 
        child { node (a) {$a$} }
        child { node (f2) {$\treeOpRecurOr_{f}$} }
      }
    };
    
  \begin{pgfonlayer}{background}
    \node[sel,fit=(xor)(b)(seq)(a)(f2)] (highlight) {};
  \end{pgfonlayer}
\end{tikzpicture}

%% file: gfx/disc/ex3-rd-s3-log.tex
\begin{tabular}{|c|c|}
    \hline
    $f$     & $f$ \\
    \hline
    \cellcolor{ELsel}$a$     & \cellcolor{ELsel}$f$ \\
    \hline
            & \cellcolor{ELsel}$b$ \\
    \hline
\end{tabular}%

%% file: gfx/disc/alg-recurse-rad.tex
\vspace*{-0.8ex}%
\begin{algorithm}
    \algfontsize
    \caption{Recursion Aware Discovery (\emph{RAD})\label{alg:recurse-discovery}}
    \alginout
        {A hierarchical event log $\symEventLog$}
        {A hierarchical process tree $\symTree$ such that $\symEventLog$ fits
        $\symTree$}
    \algdescript{Extended framework, using the named subtree and recursion operators.}
    \algname{$RAD$}{$\symEventLog$}
\begin{algtab}
    \emph{// discover root model using the full event log ($\symContextPath = \symEmptyTrace$)} \\
    \label{alg:recurse-discovery:root}$\mathit{root} = $ \algcall{$RADrun$}{$\symEventLog ,\, \symEmptyTrace$} \vspace*{0.2em}\\
    
    \emph{// discover the submodels using the recorded sublogs ($\symContextPath \neq \symEmptyTrace$)} \\
    \textbf{Let} $\mathit{model}$ be an empty map, relating context paths to process trees \\
    \label{alg:recurse-discovery:iterate}\algwhiledo{$\exists \symContextPath \in \setseq{\actAlph} : \symEventLog(\symContextPath)$ changed}{
        $\mathit{model}(\symContextPath) = $ \algcall{$RADrun$}{$\symEventLog(\symContextPath),\, \symContextPath$}
     \vspace*{0.2em} }
    
    \emph{// glue the partial models~$\mathit{model}(\symContextPath)$ and root model~$\mathit{root}$ together} \\
    \algforeach{node $P$ in process tree $\mathit{root}$ (any order, including new children)}
        Let $\symContextPath = \seqc{ f }{ P' = \treeOpCompOr_f \text{ \textbf{foreach} $P'$ on the path from $\mathit{root}$ to $P$ } }$ \\
        \label{alg:recurse-discovery:glue}\algifthen{$(\exists f : P = \treeOpCompOr_f) \land \symContextPath \in \mathit{model}$}{
            Set $\mathit{model}(\symContextPath)$ as the child of $P$ 
        }
    \algend
    \algreturn $\mathit{root}$
\end{algtab}
\end{algorithm}
\vspace*{-0.8ex}%

%% file: gfx/disc/alg-recurse-radrun.tex
\vspace*{-0.8ex}%
\begin{algorithm}
    \algfontsize
    \caption{Recursion Aware Discovery - single run\label{alg:recurse-discovery-run}}
    \alginout
        {A hierarchical event log $\symEventLog$, and a context path
        $\symContextPath$}
        {A hierarchical process tree $\symTree$ such that $\symEventLog$ fits
        $\symTree$}
    \algdescript{One single run/iteration in the RAD extended framework.}
    \algname{$RADrun$}{$\symEventLog ,\, \symContextPath$}
\begin{algtab}
  \algif{$\forall \symTrace \in \symEventLog : \symTrace = \symEmptyTrace$}
    \algreturn $\actSilent$ \hspace*{1em}\emph{// the log is empty or only contains empty traces} \\
  \algelseif{$\exists f \in \actAlph : \forall \symTrace \in \symEventLog : \symTrace = \seq{f}$}
    \algreturn $f$ \hspace*{1em}\emph{// the log only has a single low-level activity} \\
    \algelseif{$\exists f \in \symContextPath : \forall x \in \symTrace \in \symEventLog : x[1] = f$\label{alg:recurse-discovery-run:check-recursion}}
    \emph{// recursion on $f$ is detected} \\
    $\symContextPath' = \symContextPath_1 \seqcat \seq{f}$ 
        \textbf{where} $(\symContextPath_1 \seqcat \seq{f} \seqcat
        \symContextPath_2) = \symContextPath$ \\
    \label{alg:recurse-discovery-run:recursion}$\symEventLog(\symContextPath') = \symEventLog(\symContextPath') \cup \seqproject{\symEventLog}{1}$ $\quad$ \emph{// $\seqproject{\symEventLog}{1}$ is added to the sublog for $\symContextPath'$} \\
    \algreturn $\treeOpRecurOr_f$ \\
  \algelseif{$\exists f \in \actAlph: (\forall x \in t \in \symEventLog : x[1] = f) \land (\exists x \in t \in \symEventLog : \len{x} > 1)$}
      \emph{// discovered a named subtree $f$, note that $f \notin \symContextPath$ since line~\ref{alg:recurse-discovery-run:check-recursion} was false} \\
      \label{alg:recurse-discovery-run:subtree}$\symEventLog(\symContextPath \seqcat \seq{f}) = \seqproject{\symEventLog}{1}$ $\qquad$ \emph{// $\seqproject{\symEventLog}{1}$ is associated with 
          $\symContextPath' = \symContextPath \seqcat \seq{x}$} \\
      \algreturn $\treeOpCompOr_f$ \\
  \algelse
    \emph{// normal framework cases, based on the tree semantics (Def.~\ref{def:processtree})} \\
    \label{alg:recurse-discovery-run:opsplit}Split $L$ into into sublogs $\symEventLog_1, \ldots, \symEventLog_n$, such
    that: \\
    $\exists \treeOp \in \set{\treeOpSeq, \treeOpXor, \treeOpLoop, \treeOpPar}
    : L \subseteq \treeOp^l(L_1, \ldots, L_n)$ \hspace*{1em}\emph{// see~\cite{sleemans-thesis}}\\
    \algreturn $\treeOp(\funn{RADrun}(\symEventLog_1, \symContextPath), \ldots,
    \funn{RADrun}(\symEventLog_n, \symContextPath))$ \\
  \algend
\end{algtab}
\end{algorithm}
\vspace*{-0.8ex}%

%% file: gfx/reduct/reduct-pt.tex
\begin{table}[!htb]
  \footnotesize
  \centering
  \caption{Reduction rules for (hierarchical) process trees}
    \vspace*{-0.2cm}
   \begin{tabular}{@{\hskip0pt}>{\raggedleft\arraybackslash}m{28mm}@{\hskip3pt}
   >{\arraybackslash}m{28mm}@{\hskip0pt}
   >{\arraybackslash}m{28mm}@{\hskip0pt}} 
        $\treeOp(\symTree_1)$ 
      & $= \symTree_1$
      & \textbf{ for } $\treeOp \in \set{\treeOpSeq, \treeOpXor, \treeOpPar}$
      \\
        $\treeOp(\ldots_1, \treeOp(\ldots_2), \ldots_3)$ 
      & $= \treeOp(\ldots_1, \ldots_2, \ldots_3)$
      & \textbf{ for } $\treeOp \in \set{\treeOpSeq, \treeOpPar}$
      \\
        $\treeOp(\ldots_1, \actSilent, \ldots_2)$ 
      & $= \treeOp(\ldots_1, \ldots_2)$
      & \textbf{ for } $\treeOp \in \set{\treeOpSeq, \treeOpPar}$
      \\
        $\treeOpXor(\ldots_1, \actSilent, \ldots_2)$ 
      & $= \treeOpXor(\ldots_1, \ldots_2)$
      & \textbf{ if } $\symEmptyTrace \in \treeLan(\ldots_1 \cup \ldots_2)$
      \\
   \end{tabular}
    \vspace*{-0.3cm}
  \label{tab:reduct-pt}%
\end{table}%

%% file: gfx/reduct/reduct-depth.tex
\begin{tikzpicture}[
    sibling distance = 4em,
    level distance = 2em,
    every node/.style = {draw=none, fill=none, align=center}
]
  \node (pre) {$\treeOpSeq$}
    child { node (pre_a) {$a$} }
    child { node (pre_x) {$\treeOpCompOr_x$}
      child {node (pre_seq2) {$\treeOpSeq$}
        child { node {$b$} }
        child { node (pre_y) {$\treeOpCompOr_y$}
          child { node {$c$} }
        }
      }
    }
    ;
    
  \node[right of=pre, node distance=4cm] (post) {$\treeOpSeq$}
    child { node (post_a) {$\actSilent$} }
    %child { node {$\treeOpCompOr_x$}
    %child { node (post_x) {$\,$}
      child {node (post_seq2) {$\treeOpSeq$}
        child { node {$b$} }
        child { node (post_y) {$y$}
        }
      }
    %}
    ;
    
  \node [above of=pre, node distance=0.5cm, text height=0.3cm] 
    {\footnotesize\textbf{Discovered Model}};
  \node [above of=post, node distance=0.5cm, text height=0.3cm] 
    {\footnotesize\textbf{Model after Depth Filtering}};
    
  \draw[dashed,->,gray]  (pre_a.east) to[bend left=20] (post_a.west);
  %\draw[dashed,->]  (pre_x.east) to[bend right=20] (post_x.west);
  \draw[dashed,->,gray]  (pre_x.east) to[bend left=20] (post_seq2.west);
  %\draw[dashed,->]  (pre_seq2.east) to[bend right=20] (post_seq2.west);
  \draw[dashed,->,gray]  (pre_y.east) to[bend right=20] (post_y.west);
  
  \node[anchor=east,gray] at (-1.1,-1.07) {\footnotesize\emph{min depth}};
  \node[anchor=east,gray] at (-1.1,-2.5) {\footnotesize\emph{max depth}};
  \draw[dashed,gray] (-1.1,-1.07) -- (1.8,-1.07);
  \draw[dashed,gray] (-1.1,-2.5) -- (1.8,-2.5); 
  
\end{tikzpicture}

%% file: gfx/map2/pt.tex
\begin{tikzpicture}[
    sibling distance = 4em,
    level distance = 2.5em,
    every node/.style = {draw=none, fill=none, align=center}
]
  \node{$\treeOpCompOr_{\text{\cf{Main.main()}}}$}
    child { node{$\treeOpSeq$}
      child { node{\cf{Main.input()}} }
      child { node{$\treeOpXor$}
        child {node{\cf{A.f()}} }
        child {node{\cf{B.f()}} }
      }
    }
    ;
  \node at (0.0, -3.1) {};
\end{tikzpicture}

%% file: gfx/map2/sc.tex
% tex.stackexchange.com/questions/78172/statecharts-using-tikz
\begin{tikzpicture}[
    round/.style={rounded corners=1.5mm,minimum width=1cm,inner sep=2mm,above
right,draw,align=center},
    edge/.style={-latex, semithick}
]
% Nodes
  \node[circle, draw, scale=1.5] (start) at (1.9,0.45) {};
  \node[circle, draw, scale=1.5, double] (end) at (1.9,-3.05) {};
  
  \node[round, text width=18mm] (input) at (0.8,-0.95) {\cf{Main.input()}};
  
  \node[circle, fill=black, scale=0.5] (xStart) at (1.9,-1.3) {};
  \node[circle, fill=black, scale=0.5] (xEnd) at (1.9,-2.3) {};
  
  \node[round] (A) at (0.4,-2.1) {\cf{A.f()}};
  \node[round] (B) at (2.2,-2.1) {\cf{B.f()}};
  
  \node[inner sep=0mm] (lblMain) at (0.8,-0.15) {\cf{Main.main()}};
  \node[round,fit=(input)(A)(B)(lblMain)(xStart)(xEnd),inner sep=1.4mm] (Main) {};
  
% Edges
  \draw[edge] (start) to (input);
  \draw[edge] (input) to (xStart);
  \draw[edge] (xStart) to (A);
  \draw[edge] (xStart) to (B);
  \draw[edge] (A) to (xEnd);
  \draw[edge] (B) to (xEnd);
  \draw[edge] (xEnd) to (end);
  
  % AND split line
  % \draw[densely dashed] (Process.north) -- (Process.south);
  
% Symbols
  % Start symbol
  \scope[shift={(1.72,0.64)},yscale=0.15mm, xscale=0.15mm]
  \path[y=0.80pt, x=0.80pt, yscale=-1.000000, xscale=1.000000, inner sep=0pt,
  outer sep=0pt, fill=black] (8.0000,4.0000) -- (8.0000,28.0000) --
  (26.0000,16.0000) -- cycle;
  \endscope
  
  % End symbol
  \scope[shift={(1.79,-2.94)},yscale=0.16mm, xscale=0.16mm]
  \scope[y=0.80pt, x=0.80pt, yscale=-1.000000, xscale=1.000000, inner sep=0pt,
    outer sep=0pt]
  \scope[shift={(-96.0,-480.0)},fill=black]
  \path[fill] (96.0000,480.0000) -- (96.0000,496.0000) --
    (112.0000,496.0000) -- (112.0000,480.0000) -- (96.0000,480.0000) --
    cycle(96.0000,480.0000);
  \endscope
  \endscope
  \endscope
  
\end{tikzpicture}

%% file: gfx/map2/sd.tex
\begin{tikzpicture}[
    environment/.style={circle,fill,draw, scale=0.2mm},
    lifeline/.style={rectangle,draw,align=center},
    lifeend/.style={},
    linebase/.style={semithick,dash pattern={on 3pt off 3pt}},
    activation/.style={rectangle,draw,align=center, fill=white},
    fragment/.style={rectangle,draw,align=center},
    fragmentLabel/.style={align=center,inner sep=0mm},
    fragmentSep/.style={semithick,dash pattern={on 6pt off 3pt}},
    fragmentLabelBox/.style={rectangle,draw,align=center},
    message/.style={-latex, semithick},
    messageRtn/.style={-latex, semithick,dash pattern={on 3pt off 3pt}},
    msgLabel/.style={above=0.5mm, midway,fill=white,inner sep=0.5mm}
]
% Lifelines
  \node[environment] (env) at (-1.4, -0.6) {};
  \node[environment] (env_end) at (-1.4, -4.3) {};
  
  \node[lifeline] (main_s) at (0, 0.1) {\cf{Main}};
  \node[lifeend] (main_e) at (0, -4.7) {};
  \draw[linebase] (main_s) to (main_e);
    
  \node[lifeline] (a_s) at (1.0, 0.1) {\cf{A}};
  \node[lifeend] (a_e) at (1.0, -4.7) {};
  \draw[linebase] (a_s) to (a_e);
  
  \node[lifeline] (b_s) at (2.0, 0.1) {\cf{B}};
  \node[lifeend] (b_e) at (2.0, -4.7) {};
  \draw[linebase] (b_s) to (b_e);

% Activations
  \node[activation, minimum height=37mm] (actMain) at (0.0, -2.45) {};
  \node[activation, minimum height=5mm] (actA) at (1.0, -2.15) {};
  \node[activation, minimum height=5mm] (actB) at (2.0, -3.45) {};
  
% Fragments
  \node[fragmentLabel] (lblAlt) at (-0.7,-1.55) {\scriptsize{alt}};
  \node[fragmentLabelBox,fit=(lblAlt)] (lblBoxAlt) {};
  \node[fragment,fit=(lblAlt)(actA)(actB)] (AltFragment) {};
  \draw[fragmentSep] (-0.9, -2.6) to (2.2, -2.6);

% Messages
  \draw[message] (env) -- (-0.1, -0.6) node[msgLabel] {\cf{main()}};
  \draw[message] (-0.1, -4.3) -- (env_end) node[msgLabel] {\cf{main()}};
  
  \draw[message] (0.1, -0.8) arc(90:-90:0.4 and 0.15);
  \node[fill=white,inner sep=0.5mm] at (1.1, -0.93) {\cf{input()}};
  
  \draw[message] (0.1,-1.9) -- (0.9, -1.9) node[msgLabel] {\cf{f()}};
  \draw[messageRtn] (0.9, -2.4) -- (0.1,-2.4) node[msgLabel] {\cf{f()}};
  
  \draw[message] (0.1,-3.2) -- (1.9, -3.2) node[msgLabel] {\cf{f()}};
  \draw[messageRtn] (1.9, -3.7) -- (0.1,-3.7) node[msgLabel] {\cf{f()}};
\end{tikzpicture}

%% file: gfx/map2/ptnet.tex
\begin{tikzpicture}[
    place/.style={circle,draw},
    trans/.style={rectangle,draw,align=center},
    ttau/.style={rectangle,draw,fill,align=center,text width=0.4mm, 
    text height=2mm, ,minimum width=0.4mm},
    edge/.style={-latex, semithick}
]
% Nodes
  \node[place,fill=black,scale=0.4] (token) at (0, 0) {};
  \node[place] (ps) at (0, 0) {};
  \node[trans] (main_s) at (0, -0.7) {\cf{Main.main()+start}};
  \node[place] (p1) at (2.0, -1.2) {};
  
  \node[trans] (in_s) at (0, -1.3) {\cf{Main.input()+start}};
  \node[place] (pi) at (2.0, -1.9) {};
  \node[trans] (in_c) at (0, -1.9) {\cf{Main.input()+end}};
  \node[place] (p2) at (0, -2.55) {};
   
  \node[trans] (af_s) at (-1.4, -2.55) {\cf{A.f()+start}};
  \node[place] (pa2) at (-1.4, -3.2) {};
  \node[trans] (af_c) at (-1.4, -3.85) {\cf{A.f()+end}};
   
  \node[trans] (bf_s) at (1.4, -2.55) {\cf{B.f()+start}};
  \node[place] (pb2) at (1.4, -3.2) {};
  \node[trans] (bf_c) at (1.4, -3.85) {\cf{B.f()+end}};
   
  \node[place] (p3) at (0, -3.85) {};
  
  \node[trans] (main_c) at (0, -4.5) {\cf{Main.main()+end}};
  \node[place] (pc) at (0, -5.2) {};
  
% Edges
  \draw[edge] (ps) to (main_s);
  \draw[edge] (main_s) to (2.0, -0.7) to (p1);
  
  \draw[edge] (p1) to (1.45, -1.2);
  \draw[edge] (1.45, -1.45) to (2.0, -1.45) to (pi);
  \draw[edge] (pi) to (1.3, -1.9);
  \draw[edge] (in_c) to (p2);
  
   \draw[edge] (p2) to (af_s);
   \draw[edge] (af_s) to (pa2);
   \draw[edge] (pa2) to (af_c);
   \draw[edge] (af_c) to (p3);
   
   \draw[edge] (p2) to (bf_s);
   \draw[edge] (bf_s) to (pb2);
   \draw[edge] (pb2) to (bf_c);
   \draw[edge] (bf_c) to (p3);
  
  \draw[edge] (p3) to (main_c);
  \draw[edge] (main_c) to (pc);
\end{tikzpicture}

%% file: gfx/eval2/log-info.tex
\begin{table}[htbp]
  \centering
\begin{threeparttable}%
  \caption{The event logs used in the evaluation, with input sizes}
  \begin{tabular}{@{\hskip3pt}c@{\hskip3pt}l |
      D{.}{.}{6.0} 
      D{.}{.}{7.0}  
      D{.}{.}{3.0}  
      D{.}{.}{6.2}
    }
    & Event Log 
    & \multicolumn{1}{@{\hskip4pt} c @{\hskip4pt}}{\# Traces}
    & \multicolumn{1}{@{\hskip4pt} c @{\hskip4pt}}{\# Events}
    & \multicolumn{1}{@{\hskip4pt} c @{\hskip4pt}}{\# Acts}
    & \multicolumn{1}{@{\hskip4pt} c @{\hskip4pt}}{Avg. $|$Trace$|$} \\
   \midrule
   \cite{xeslog:bpic2012} & BPIC 2012
     & 13,087 & 262,200  & 24 & 20.04 \\
   \cite{xeslog:bpic2013} & BPIC 2013
     & 7,554 & 65,533 & 13 & 8.70 \\
   \cite{xeslog:junit412} & JUnit 4.12
     & 1 & 946 & 182 & 946.00 \\
   \cite{xeslog:apache-crypto} & Crypto 1.0.0
     & 3 & 241,973 & 74 & 80,657.67 \\
   \cite{xeslog:nasa-cev} & NASA CEV
     & 2,566 & 73,638 & 47 & 28.70 \\
  \end{tabular}%
  \label{tab:eval2:logs}%
  \end{threeparttable}%
\end{table}%

%% file: gfx/eval3/result-time.tex
\definecolor{plotcoloraxis}{rgb}{0.8,0.8,0.8}
\definecolor{plotcolorfill}{rgb}{0.7,0.7,0.7}

\newcommand{\evalT}{-\tnote{T}}
\newcommand{\evalNA}{\tnote{n/a}}
\newcommand{\evalM}{-\tnote{M}}

%\multirow{14}{2.7cm}{\hspace*{8ex}%
\newcommand{\plottime}[1]{
\multirow{23}{1.4cm}{\hspace*{-3ex}
\begin{tikzpicture}[trim axis left,trim axis right]
\begin{axis}[
  % axis setup
  y=1.052*\baselineskip,
  width=1.5cm,
  scale only axis,
  hide y axis,
  axis line style={opacity=0},
  major tick style={draw=none},
  minor tick style={draw=none},
  xmajorgrids,
  %major grid style=plotcoloraxis,
  major grid style=white,
  axis on top,
  enlarge y limits={abs=0.4},
  %axis y line*=middle,
  %ytick=\empty,
  axis x line*=bottom,
  % bar plot setup
  xbar,
  bar width=1.5ex,
  xmin=0,
  x tick label style={
    rotate=45,
    anchor=east
  },
  scaled ticks=false,
  xmode=log,
  % data label nodes
  %visualization depends on=x \as \rawx,
  %nodes near coords,
  %every node near coord/.style={
  %  anchor=east,
  %  shift={(axis direction cs:-\rawx,0)},
  %  color=black
  %},
  %xlabel={\scriptsize (ms) [95\% conf.]},
  %every axis x label/.style={
  %  anchor=south,
  %  shift={(axis direction cs:100,-2.7)}
  %}
  ]
\addplot+[
  white,
  fill=plotcolorfill
  ][
  error bars/.cd,
  x dir=both, 
  x explicit,
  error bar style={black}
  ] coordinates {
#1
};
\end{axis}
\end{tikzpicture}%
}
\hspace*{-4ex}
}

\begin{table*}[htbp]
  \centering
\noindent\adjustbox{max width=0.99\textwidth}{%
\begin{threeparttable}%
  \caption{
    Running time for the different algorithms, hierarchy heuristics, paths frequency filter settings, and event logs. 
  }
    \begin{tabular}{
    @{\hskip3pt}c@{\hskip3pt}l@{\hskip4pt}c
    |
    D{.}{.}{1}c 
    D{.}{.}{1}c 
    D{.}{.}{1}c 
    D{.}{.}{1}c 
    D{.}{.}{1}c
    }
          & Algorithm (Heuristic) & Paths 
          & \multicolumn{2}{c}{BPIC 2012} 
          & \multicolumn{2}{c}{BPIC 2013} 
          & \multicolumn{2}{c}{JUnit 4.12} 
          & \multicolumn{2}{c}{Crypto 1.0.0} 
          & \multicolumn{2}{c}{NASA CEV} \\
          \midrule
    \cite{aalst2004workflow-alpha} & Alpha miner &       
    & 150.1 & \plottime{(150.0705,22.5)+-(5.2598,-5.2598) (840.169,21.5)+-(4.4131,-4.4131) (2858.5425,20.5)+-(16.2253,-16.2253) (0,19.5)+-(0,0) (7234.3495,18.5)+-(532.4335,-532.4335) (0,17.5)+-(0,0) (0,16.5)+-(0,0) (0,15.5)+-(0,0) (0,14.5)+-(0,0) (0,13.5)+-(0,0) (0,12.5)+-(0,0) (0,11.5)+-(0,0) (3239.9375,10.5)+-(29.6637,-29.6637) (5111.6387,9.5)+-(54.8337,-54.8337) (3588.057,8)+-(28.9731,-28.9731) (2865.5773,7)+-(41.1475,-41.1475) (0,6)+-(0,0) (0,5)+-(0,0) (0,4)+-(0,0) (0,3)+-(0,0) (2058.5896,2)+-(92.5553,-92.5553) (2394.9674,1)+-(259.932,-259.932) } 
    & 73.5  & \plottime{(73.5357,22.5)+-(1.064,-1.064) (278.0308,21.5)+-(10.3514,-10.3514) (827.4219,20.5)+-(4.3833,-4.3833) (6023.7899,19.5)+-(101.1957,-101.1957) (4354.7201,18.5)+-(42.8913,-42.8913) (0,17.5)+-(0,0) (0,16.5)+-(0,0) (0,15.5)+-(0,0) (0,14.5)+-(0,0) (0,13.5)+-(0,0) (0,12.5)+-(0,0) (0,11.5)+-(0,0) (1351.2428,10.5)+-(12.7801,-12.7801) (947.4318,9.5)+-(7.7122,-7.7122) (1117.8147,8)+-(9.5315,-9.5315) (603.3966,7)+-(6.4494,-6.4494) (0,6)+-(0,0) (0,5)+-(0,0) (0,4)+-(0,0) (0,3)+-(0,0) (891.0974,2)+-(17.032,-17.032) (1028.645,1)+-(92.3924,-92.3924) } 
    & 9.2   & \plottime{(9.2177,22.5)+-(2.2602,-2.2602) (1349.6518,21.5)+-(4.988,-4.988) (166.7606,20.5)+-(0.7451,-0.7451) (0,19.5)+-(0,0) (0,18.5)+-(0,0) (243.9321,17.5)+-(3.799,-3.799) (581.9946,16.5)+-(3.0398,-3.0398) (751.7655,15.5)+-(1.4915,-1.4915) (108.9328,14.5)+-(1.0461,-1.0461) (371.6499,13.5)+-(0.6436,-0.6436) (512.2808,12.5)+-(0.7281,-0.7281) (0,11.5)+-(0,0) (215.6757,10.5)+-(3.4474,-3.4474) (268.5292,9.5)+-(3.7121,-3.7121) (278.3296,8)+-(4.0439,-4.0439) (298.4684,7)+-(2.8323,-2.8323) (15.0682,6)+-(0.2549,-0.2549) (12.5703,5)+-(0.0979,-0.0979) (16.4387,4)+-(0.306,-0.306) (14.5567,3)+-(0.1483,-0.1483) (23.2945,2)+-(0.2213,-0.2213) (23.6531,1)+-(0.1413,-0.1413) } 
    & 183.1 & \plottime{(183.1289,22.5)+-(0.5473,-0.5473) (0,21.5)+-(0,0) (0,20.5)+-(0,0) (0,19.5)+-(0,0) (0,18.5)+-(0,0) (0,17.5)+-(0,0) (0,16.5)+-(0,0) (0,15.5)+-(0,0) (0,14.5)+-(0,0) (0,13.5)+-(0,0) (0,12.5)+-(0,0) (0,11.5)+-(0,0) (10866.1318,10.5)+-(379.3743,-379.3743) (5213.9391,9.5)+-(104.1891,-104.1891) (26804.7475,8)+-(51.5115,-51.5115) (0,7)+-(0,0) (1544.6311,6)+-(22.5103,-22.5103) (1545.6299,5)+-(34.7483,-34.7483) (2186.3708,4)+-(698.88,-698.88) (2082.8396,3)+-(533.047,-533.047) (0,2)+-(0,0) (1,1)+-(0,0) } 
    & 37.8  & \plottime{(37.7659,22.5)+-(0.6964,-0.6964) (359.6048,21.5)+-(2.9378,-2.9378) (4148.231,20.5)+-(8.9145,-8.9145) (0,19.5)+-(0,0) (0,18.5)+-(0,0) (13426.0167,17.5)+-(53.7755,-53.7755) (22213.3887,16.5)+-(33.5082,-33.5082) (0,15.5)+-(0,0) (0,14.5)+-(0,0) (0,13.5)+-(0,0) (0,12.5)+-(0,0) (0,11.5)+-(0,0) (911.2937,10.5)+-(14.0079,-14.0079) (912.5723,9.5)+-(17.3952,-17.3952) (959.4593,8)+-(11.5672,-11.5672) (1047.0582,7)+-(13.8196,-13.8196) (355.8747,6)+-(4.3709,-4.3709) (341.4533,5)+-(4.4522,-4.4522) (438.9791,4)+-(95.2008,-95.2008) (373.8024,3)+-(35.6501,-35.6501) (1275.901,2)+-(34.4305,-34.4305) (1452.3305,1)+-(155.3607,-155.3607) } \\
    \cite{Weijters2011-fhm} & Heuristics &       & 840.2 &       & 278.0 &       & 1349.7 &       & \evalT &       & 359.6 &  \\
    \cite{gunther2007fuzzy} & Fuzzy miner &       & 2858.5 &       & 827.4 &       & 166.8 &       & \evalT &       & 4148.2 &  \\
    \cite{vanderWerf2008-ilp} & ILP miner &       & \evalT &       & 6023.8 &       & \evalT &       & \evalT &       & \evalT &  \\
    \cite{zelst2015-ilp-filter} & ILP with filtering &       & 7234.3 &       & 4354.7 &       & \evalT &       & \evalT &       & \evalT &  \\
    \cite{Walkinshaw2016mint} & MINT, redblue, k=1 &       & \evalT &       & \evalT &       & 243.9 &       & \evalT &       & 13426.0 &  \\
    \cite{Walkinshaw2016mint} & MINT, redblue, k=2 &       & \evalT &       & \evalT &       & 582.0 &       & \evalT &       & 22213.4 &  \\
    \cite{Walkinshaw2016mint} & MINT, redblue, k=3 &       & \evalT &       & \evalT &       & 751.8 &       & \evalT &       & \evalT &  \\
    \cite{Walkinshaw2016mint} & MINT, ktails, k=1 &       & \evalT &       & \evalT &       & 108.9 &       & \evalT &       & \evalT &  \\
    \cite{Walkinshaw2016mint} & MINT, ktails, k=2 &       & \evalT &       & \evalT &       & 371.6 &       & \evalT &       & \evalT &  \\
    \cite{Walkinshaw2016mint} & MINT, ktails, k=3 &       & \evalT &       & \evalT &       & 512.3 &       & \evalT &       & \evalT &  \\
    \cite{Beschastnikh2011synoptic} & Synoptic &       & \evalT &       & \evalT &       & \evalT &       & \evalT &       & \evalT &  \\
    \cite{sleemans-thesis} & IM (baseline) & 1.0   & 3239.9 &       & 1351.2 &       & 215.7 &       & 10866.1 &       & 911.3 &  \\
    \cite{sleemans-thesis} & IM (baseline) & 0.8   & 5111.6 &       & 947.4 &       & 268.5 &       & 5213.9 &       & 912.6 &  \\
    \midrule
    \multirow{8}[0]{*}{\rotatebox[origin=c]{90}{\textbf{Our techniques}}}  & Na\"ive (no heuristic) & 1.0   & 3588.1 &       & 1117.8 &       & 278.3 &       & 26804.7 &       & 959.5 &  \\
    \multicolumn{1}{c}{} & Na\"ive (no heuristic) & 0.8   & 2865.6 &       & 603.4 &       & 298.5 &       & \evalT &       & 1047.1 &  \\
    \multicolumn{1}{c}{} & Na\"ive (Nested Calls) & 1.0   & \evalNA &       & \evalNA &       & 15.1  &       & 1544.6 &       & 355.9 &  \\
    \multicolumn{1}{c}{} & Na\"ive (Nested Calls) & 0.8   & \evalNA &       & \evalNA &       & 12.6  &       & 1545.6 &       & 341.5 &  \\
    \multicolumn{1}{c}{} & RAD (Nested Calls) & 1.0   & \evalNA &       & \evalNA &       & 16.4  &       & 2186.4 &       & 439.0 &  \\
    \multicolumn{1}{c}{} & RAD (Nested Calls) & 0.8   & \evalNA &       & \evalNA &       & 14.6  &       & 2082.8 &       & 373.8 &  \\
    \multicolumn{1}{c}{} & Na\"ive (Struct. Names) & 0.8   & 2058.6 &       & 891.1 &       & 23.3  &       & \evalM &       & 1275.9 &  \\
    \multicolumn{1}{c}{} & RAD (Struct. Names) & 0.8   & 2395.0 &       & 1028.6 &       & 23.7  &       & \evalM &       & 1452.3 &  \\
    \multicolumn{13}{r}{} \\[0.7em]
    \multicolumn{3}{r}{} & \multicolumn{10}{c}{Avg. runtime (in milliseconds, with log scale plot), over 30 runs, with 95\% confidence interval}
    \end{tabular}%
    \vspace*{-1.4em}
    %\vspace*{-1em}
    \begin{tablenotes}%
    \begin{multicols}{3}
    \item[M] Out of memory exception (12 GB)
    \item[T] Time limit exceeded (30 sec.)
    \item[n/a] Heuristic not applicable
     \end{multicols}
    \end{tablenotes}%
  \label{tab:eval2:results:time}%
  \end{threeparttable}%
}%
    \vspace*{-2.0em}
\end{table*}%

%% file: gfx/eval2/log-hierarchy.tex
\renewcommand{\evalNA}{\tnote{n/a}\hspace*{0.85em}}

\begin{table}[htp!]
  \centering
  \caption{
    Depth of the discovered hierarchy, 
    for the different event logs, hierarchy heuristics, and algorithms}
    \vspace*{-1em}
\noindent\adjustbox{max width=0.45\textwidth}{%
\begin{threeparttable}%
  \begin{tabular}{@{\hskip3pt}c@{\hskip3pt}l |
      c @{\hskip1pt}
      D{.}{.}{3.0}
      D{.}{.}{3.0}
      c @{\hskip1pt}
      D{.}{.}{3.0}
      D{.}{.}{3.0}
    }
    & \multicolumn{1}{ c }{}
    & 
    & \multicolumn{2}{ c }{Nested Calls}
    & 
    & \multicolumn{2}{ c }{Struct. Names}
    \\
    \cmidrule{4-5}
    \cmidrule{7-8}
    & Event Log 
    & 
    & \multicolumn{1}{ c }{Na\"ive}
    & \multicolumn{1}{ c }{RAD}
    & 
    & \multicolumn{1}{ c }{Na\"ive}
    & \multicolumn{1}{ c }{RAD}
    \\
   \midrule
   \cite{xeslog:bpic2012} & BPIC 2012 &
     & \evalNA & \evalNA & & 2 & 2 \\
   \cite{xeslog:bpic2013} & BPIC 2013 &
     & \evalNA & \evalNA & & 2 & 2 \\
   \cite{xeslog:junit412} & JUnit 4.12 &
     & 25 & 18 & & 9 & 9 \\
   \cite{xeslog:apache-crypto} & Apache Crypto 1.0.0 &
     & 8 & 8 & & \evalNA & \evalNA \\
   \cite{xeslog:nasa-cev} & NASA CEV &
     & 3 & 3 & & 3 & 3 \\
  \end{tabular}%
    %\vspace*{-1.5em}
    \vspace*{-1em}
    \begin{tablenotes}%
    \begin{multicols}{2}
    \item $\,$
    \item[n/a] No model (see Table~\ref{tab:eval2:results:time})
     \end{multicols}
    \end{tablenotes}%
  \label{tab:eval2:results:depth}%
  \end{threeparttable}%
}%
    \vspace*{-2em}
\end{table}%

%% file: gfx/eval3/result-quality.tex
\definecolor{plotcoloraxis}{rgb}{0.8,0.8,0.8}
\definecolor{plotcolorfill}{rgb}{0.7,0.7,0.7}

\renewcommand{\evalT}{-\tnote{T}}
\newcommand{\evalU}{-\tnote{U}}
\newcommand{\evalN}{\tnote{n/a}}
\newcommand{\evalR}{-\tnote{R}}

%\multirow{14}{2.7cm}{\hspace*{8ex}%
\newcommand{\plotquality}[1]{
\multirow{19}{0.4cm}{\hspace*{-3.5ex}
\begin{tikzpicture}[trim axis left,trim axis right]
\begin{axis}[
  % axis setup
  y=1.05*\baselineskip,
  width=0.5cm,
  scale only axis,
  hide y axis,
  axis line style={opacity=0},
  major tick style={draw=none},
  %xmajorgrids,
  %major grid style=plotcoloraxis,
  %major grid style=white,
  %axis on top,
  enlarge y limits={abs=0.4},
  %axis y line*=middle,
  %ytick=\empty,
  axis x line*=bottom,
  % bar plot setup
  xbar,
  bar width=1.5ex,
  xmin=0.0,
  xmax=1.0,
  ticks=none,
  scaled ticks=false
  % data label nodes
  %visualization depends on=x \as \rawx,
  %nodes near coords,
  %every node near coord/.style={
  %  anchor=east,
  %  shift={(axis direction cs:-\rawx,0)},
  %  color=black
  %},
  %xlabel={\scriptsize (ms) [95\% conf.]},
  %every axis x label/.style={
  %  anchor=south,
  %  shift={(axis direction cs:100,-2.7)}
  %}
  ]
\addplot+[
  white,
  fill=plotcolorfill
  ][] coordinates {
#1
};
\end{axis}
\end{tikzpicture}%
}
\hspace*{-4ex}
}

\begin{table*}[htp!]
  \centering
\noindent\adjustbox{max width=0.99\textwidth}{%
\begin{threeparttable}%
  \caption{
    Model quality scores for the different algorithms, hierarchy heuristics, paths frequency filter settings, and event logs.
    Scores range from 0.0 to 1.0, higher is better. No scores are available for the Fuzzy model and Structured Names hierarchy.
  }
    \begin{tabular}{
    @{\hskip3pt}c@{\hskip3pt}l@{\hskip5pt}c
    |
    @{\hskip3pt}c@{\hskip3pt}
    D{.}{.}{2}c@{\hskip3pt}D{.}{.}{2}c @{\hskip3pt}c@{\hskip3pt}
    D{.}{.}{2}c@{\hskip3pt}D{.}{.}{2}c @{\hskip3pt}c@{\hskip3pt}
    D{.}{.}{2}c@{\hskip3pt}D{.}{.}{2}c @{\hskip3pt}c@{\hskip3pt}
    D{.}{.}{2}c@{\hskip3pt}D{.}{.}{2}c @{\hskip3pt}c@{\hskip3pt}
    D{.}{.}{2}c@{\hskip3pt}D{.}{.}{2}c
    }
          &  & \multicolumn{1}{ c }{}
          &
          & \multicolumn{4}{c}{BPIC 2012}
          &
          & \multicolumn{4}{c}{BPIC 2013}
          & 
          & \multicolumn{4}{c}{JUnit 4.12} 
          &
          & \multicolumn{4}{c}{Crypto 1.0.0}
          & 
          & \multicolumn{4}{c}{NASA CEV} \\
          \cmidrule{5-8}
          \cmidrule{10-13}
          \cmidrule{15-18}
          \cmidrule{20-23}
          \cmidrule{25-28}
          & Algorithm (Heuristic) & Paths 
          &
          & \multicolumn{2}{c}{\scriptsize{Fitness}}
          & \multicolumn{2}{c}{\scriptsize{Precision}} 
          &
          & \multicolumn{2}{c}{\scriptsize{Fitness}}
          & \multicolumn{2}{c}{\scriptsize{Precision}}
          &
          & \multicolumn{2}{c}{\scriptsize{Fitness}}
          & \multicolumn{2}{c}{\scriptsize{Precision}}
          &
          & \multicolumn{2}{c}{\scriptsize{Fitness}}
          & \multicolumn{2}{c}{\scriptsize{Precision}}
          &
          & \multicolumn{2}{c}{\scriptsize{Fitness}}
          & \multicolumn{2}{c}{\scriptsize{Precision}} \\
          \midrule
    \cite{aalst2004workflow-alpha} & Alpha miner &       &       
    & \evalU & \plotquality{(0,19.5) (0.7201,18.5) (0,17.5) (0.7393,16.5) (0,15.5) (0,14.5) (0,13.5) (0,12.5) (0,11.5) (0,10.5) (0,9.5) (1,8.5) (0.9768,7.5) (1,6) (0.9768,5) (0,4) (0,3) (0,2) (0,1) } 
    & \evalU & \plotquality{(0,19.5) (0.9456,18.5) (0,17.5) (0.2844,16.5) (0,15.5) (0,14.5) (0,13.5) (0,12.5) (0,11.5) (0,10.5) (0,9.5) (0.3669,8.5) (0.4936,7.5) (0.3669,6) (0.4936,5) (0,4) (0,3) (0,2) (0,1) } 
    &       
    & 0.36  & \plotquality{(0.3565,19.5) (0,18.5) (1,17.5) (0.9493,16.5) (0,15.5) (0,14.5) (0,13.5) (0,12.5) (0,11.5) (0,10.5) (0,9.5) (1,8.5) (0.9478,7.5) (1,6) (0.9478,5) (0,4) (0,3) (0,2) (0,1) } 
    & 0.88  & \plotquality{(0.8826,19.5) (0,18.5) (0.3631,17.5) (0.45,16.5) (0,15.5) (0,14.5) (0,13.5) (0,12.5) (0,11.5) (0,10.5) (0,9.5) (0.6172,8.5) (0.6377,7.5) (0.6172,6) (0.6377,5) (0,4) (0,3) (0,2) (0,1) } 
    &       
    & \evalU & \plotquality{(0,19.5) (0,18.5) (0,17.5) (0,16.5) (0,15.5) (0.4828,14.5) (0.1262,13.5) (0,12.5) (0.4297,11.5) (0.1239,10.5) (0,9.5) (1,8.5) (0.8964,7.5) (1,6) (0.8964,5) (1,4) (0.8972,3) (1,2) (0.8873,1) } 
    & \evalU & \plotquality{(0,19.5) (0,18.5) (0,17.5) (0,16.5) (0,15.5) (0.1676,14.5) (0.0568,13.5) (0,12.5) (0.1632,11.5) (0.0564,10.5) (0,9.5) (0.335,8.5) (0.3002,7.5) (0.335,6) (0.3002,5) (0.844,4) (0.8667,3) (0.827,2) (0.8449,1) }
    &       
    & \evalU & \plotquality{(0,19.5) (0,18.5) (0,17.5) (0,16.5) (0,15.5) (0,14.5) (0,13.5) (0,12.5) (0,11.5) (0,10.5) (0,9.5) (1,8.5) (0.8772,7.5) (1,6) (0.8772,5) (0.998,4) (0.9879,3) (0.998,2) (0.9879,1) } 
    & \evalU & \plotquality{(0,19.5) (0,18.5) (0,17.5) (0,16.5) (0,15.5) (0,14.5) (0,13.5) (0,12.5) (0,11.5) (0,10.5) (0,9.5) (0.3548,8.5) (0.4135,7.5) (0.3548,6) (0.4135,5) (0.4532,4) (0.4514,3) (0.4532,2) (0.4514,1) } 
    &       
    & 0.91  & \plotquality{(0.9143,19.5) (0,18.5) (0,17.5) (0,16.5) (0.7909,15.5) (0.8112,14.5) (0,13.5) (0,12.5) (0,11.5) (0,10.5) (0,9.5) (1,8.5) (0.9141,7.5) (1,6) (0.9141,5) (1,4) (0.9999,3) (1,2) (0.9999,1) } 
    & 0.06  & \plotquality{(0.06,19.5) (0,18.5) (0,17.5) (0,16.5) (0.4365,15.5) (0.4535,14.5) (0,13.5) (0,12.5) (0,11.5) (0,10.5) (0,9.5) (0.5532,8.5) (0.5263,7.5) (0.5532,6) (0.5263,5) (0.8031,4) (0.8092,3) (0.8031,2) (0.8092,1) } \\
    \cite{Weijters2011-fhm} & Heuristics &       &       & 0.72  &       & 0.95  &       &       & \evalU &       & \evalU &       &       & \evalU &       & \evalU &       &       & \evalN &       & \evalN &       &       & \evalU &       & \evalU &  \\
    \cite{vanderWerf2008-ilp} & ILP miner &       &       & \evalN &       & \evalN &       &       & 1.00  &       & 0.36  &       &       & \evalN &       & \evalN &       &       & \evalN &       & \evalN &       &       & \evalN &       & \evalN &  \\
    \cite{zelst2015-ilp-filter} & ILP, filtering &       &       & 0.74  &       & 0.28  &       &       & 0.95  &       & 0.45  &       &       & \evalN &       & \evalN &       &       & \evalN &       & \evalN &       &       & \evalN &       & \evalN &  \\
    \cite{Walkinshaw2016mint} & MINT, redblue, k=1 &       &       & \evalN &       & \evalN &       &       & \evalN &       & \evalN &       &       & 0.00  &       & \evalR &       &       & \evalN &       & \evalN &       &       & 0.79  &       & 0.44  &  \\
    \cite{Walkinshaw2016mint} & MINT, redblue, k=2 &       &       & \evalN &       & \evalN &       &       & \evalN &       & \evalN &       &       & 0.48  &       & 0.17  &       &       & \evalN &       & \evalN &       &       & 0.81  &       & 0.45  &  \\
    \cite{Walkinshaw2016mint} & MINT, redblue, k=3 &       &       & \evalN &       & \evalN &       &       & \evalN &       & \evalN &       &       & 0.13  &       & 0.06  &       &       & \evalN &       & \evalN &       &       & \evalN &       & \evalN &  \\
    \cite{Walkinshaw2016mint} & MINT, ktails, k=1 &       &       & \evalN &       & \evalN &       &       & \evalN &       & \evalN &       &       & 0.00  &       & \evalR &       &       & \evalN &       & \evalN &       &       & \evalN &       & \evalN &  \\
    \cite{Walkinshaw2016mint} & MINT, ktails, k=2 &       &       & \evalN &       & \evalN &       &       & \evalN &       & \evalN &       &       & 0.43  &       & 0.16  &       &       & \evalN &       & \evalN &       &       & \evalN &       & \evalN &  \\
    \cite{Walkinshaw2016mint} & MINT, ktails, k=3 &       &       & \evalN &       & \evalN &       &       & \evalN &       & \evalN &       &       & 0.12  &       & 0.06  &       &       & \evalN &       & \evalN &       &       & \evalN &       & \evalN &  \\
    \cite{Beschastnikh2011synoptic} & Synoptic &       &       & \evalN &       & \evalN &       &       & \evalN &       & \evalN &       &       & \evalN &       & \evalN &       &       & \evalN &       & \evalN &       &       & \evalN &       & \evalN &  \\
    \cite{sleemans-thesis} & IM (baseline) & 1.0   &       & 1.00  &       & 0.37  &       &       & 1.00  &       & 0.62  &       &       & 1.00  &       & 0.34  &       &       & 1.00  &       & 0.35  &       &       & 1.00  &       & 0.55  &  \\
    \cite{sleemans-thesis} & IM (baseline) & 0.8   &       & 0.98  &       & 0.49  &       &       & 0.95  &       & 0.64  &       &       & 0.90  &       & 0.30  &       &       & 0.88  &       & 0.41  &       &       & 0.91  &       & 0.53  &  \\
    \midrule
    \multirow{6}[0]{*}{\rotatebox[origin=c]{90}{\textbf{Our techniques}}}
          & Na\"ive (no heuristic) & 1.0   &       & 1.00  &       & 0.37  &       &       & 1.00  &       & 0.62  &       &       & 1.00  &       & 0.34  &       &       & 1.00  &       & 0.35  &       &       & 1.00  &       & 0.55  &  \\
          & Na\"ive (no heuristic) & 0.8   &       & 0.98  &       & 0.49  &       &       & 0.95  &       & 0.64  &       &       & 0.90  &       & 0.30  &       &       & 0.88  &       & 0.41  &       &       & 0.91  &       & 0.53  &  \\
          & Na\"ive (Nested Calls) & 1.0   &       & \evalN &       & \evalN &       &       & \evalN &       & \evalN &       &       & 1.00  &       & 0.84  &       &       & 1.00  &       & 0.45  &       &       & 1.00  &       & 0.80  &  \\
          & Na\"ive (Nested Calls) & 0.8   &       & \evalN &       & \evalN &       &       & \evalN &       & \evalN &       &       & 0.90  &       & 0.87  &       &       & 0.99  &       & 0.45  &       &       & 1.00  &       & 0.81  &  \\
          & RAD (Nested Calls) & 1.0   &       & \evalN &       & \evalN &       &       & \evalN &       & \evalN &       &       & 1.00  &       & 0.83  &       &       & 1.00  &       & 0.45  &       &       & 1.00  &       & 0.80  &  \\
          & RAD (Nested Calls) & 0.8   &       & \evalN &       & \evalN &       &       & \evalN &       & \evalN &       &       & 0.89  &       & 0.84  &       &       & 0.99  &       & 0.45  &       &       & 1.00  &       & 0.81  &  \\
    \end{tabular}%
    \vspace*{-1.5em}
    \begin{tablenotes}%
    \begin{multicols}{4}
    \item[T] Time limit exceeded (5 min.)
    \item[R] Not reliable (fitness = 0)
    \item[U] Unsound model
    \item[n/a] No model (see Table~\ref{tab:eval2:results:time})
     \end{multicols}
    \end{tablenotes}%
  \label{tab:eval2:results:quality}%
  \end{threeparttable}%
}%
  \vspace*{-2.3em}
\end{table*}%